\documentclass[journal,onecolumn]{IEEEtran}

\ifCLASSINFOpdf
\else
\fi

\usepackage{amsthm}
\usepackage{amsfonts,amssymb,latexsym,cite}
\usepackage[cmex10]{amsmath}
\usepackage{algorithmic}
\usepackage{array}
\usepackage{bbm}
\usepackage{mathrsfs}
\usepackage{multirow}
\usepackage{verbatim}
\usepackage{color,xcolor}
\usepackage{graphicx}
\usepackage{epstopdf}
\usepackage{subcaption}
\usepackage[font={small,it}]{caption}
\usepackage[T1]{fontenc}
\usepackage{url}
\usepackage{enumerate}
\usepackage{hyperref} 
\graphicspath{{./figures/}} 
\usepackage{caption}
\usepackage{stfloats} 
\usepackage[blocks]{authblk}
\usepackage{epsfig}
\usepackage{latexsym}
\usepackage{soul}
\newtheorem{lemma}{Lemma}
\newtheorem{definition}{Definition}
\newtheorem{theorem}{Theorem}

\newtheorem{remark}{Remark}

\newcommand{\upperRomannumeral}[1]{\uppercase\expandafter{\romannumeral#1}}

\usepackage{lipsum}

\newcommand{\R}{\mathsf{R}}
\newcommand{\V}{\mathbb{V}}
\newcommand{\PP}{\mathbb{P}}
\newcommand{\DD}{\mathbb{D}}
\newcommand{\E}{\mathbb{E}}
\newcommand{\M}{\mathcal{M}}
\newcommand{\x}{\mathbf{x}}
\newcommand{\y}{\mathbf{y}}
\newcommand{\z}{\mathbf{z}}
\newcommand{\X}{\mathbf{X}}
\newcommand{\Y}{\mathbf{Y}}
\newcommand{\Z}{\mathbf{Z}}
\newcommand{\ux}{\underline{\mathbf{x}}}
\newcommand{\uy}{\underline{\mathbf{y}}}
\newcommand{\uz}{\underline{\mathbf{z}}}

\newcommand{\wyx}{W_{\kern-0.17em Y \kern -0.06em|\kern -0.1em X}}
\newcommand{\wyxn}{W^{\kern -0.06em \otimes n}_{\kern-0.17em Y \kern -0.06em|\kern -0.1em X}}
\newcommand{\wyxw}{W^{\kern -0.06em \otimes w}_{\kern-0.17em Y \kern -0.06em|\kern -0.1em X}}
\newcommand{\wzx}{W_{\kern-0.17em Z \kern -0.06em|\kern -0.1em X}}
\newcommand{\wzxn}{W^{\kern -0.06em \otimes n}_{\kern-0.17em Z \kern -0.06em|\kern -0.1em X}}
\newcommand{\wzxw}{W^{\kern -0.06em \otimes w}_{\kern-0.17em Z \kern -0.06em|\kern -0.1em X}}

\newcommand{\D}{\mathcal{D}}
\newcommand{\F}{\mathcal{F}}
\newcommand{\C}{\mathcal{C}}
\newcommand{\B}{\mathcal{B}}

\newcommand{\pwx}{P^{w}_{\underline{\mathbf{X}}}}
\newcommand{\pwy}{P^{w}_{\underline{\mathbf{Y}}}}
\newcommand{\pwz}{P^{w}_{\underline{\mathbf{Z}}}}

\newcommand{\ppmx}{P_{\X}^{n,l}}
\newcommand{\ppmy}{P_{\Y}^{n,l}}
\newcommand{\ppmz}{P_{\Z}^{n,l}}
\newcommand{\Um}{U_m}
\newcommand{\tUm}{\widetilde{U}_m}

\newcommand{\Uxm}{U_m(\x)}
\newcommand{\Uxhm}{U_{m'}(\x)}

\allowdisplaybreaks[3]

\begin{document}
%
\title{Covert Identification over  \\ Binary-Input Discrete Memoryless Channels}
%
%
%

\author{
        Qiaosheng Zhang,
        Vincent~Y.~F.~Tan,~\IEEEmembership{Senior~Member,~IEEE}
\thanks{Qiaosheng Zhang is with the Department
of Electrical and Computer Engineering, National University of Singapore (e-mail: elezqiao@nus.edu.sg).} 
\thanks{Vincent~Y.~F.~Tan is with the Department
of Electrical and Computer Engineering and Department of Mathematics, National University of Singapore (e-mail: vtan@nus.edu.sg).}
}

%
%

\markboth{}%
{Shell \MakeLowercase{\textit{et al.}}: Bare Demo of IEEEtran.cls for IEEE Journals}
%



\maketitle

\begin{abstract}
This paper considers the covert identification problem in which a sender aims to reliably convey an identification (ID) message to a set of receivers via a binary-input discrete memoryless channel (BDMC), and simultaneously to guarantee that the communication is covert with respect to a warden who monitors the communication via another independent BDMC. We prove a square-root law for the covert identification problem. This states that an ID message of size $\exp(\exp(\Theta(\sqrt{n})))$ can be transmitted over $n$ channel uses. We then characterize the exact pre-constant in the $\Theta(\cdot)$ notation. This constant  is referred to as the covert identification capacity. We show that  it equals the recently developed covert capacity in the standard covert communication problem, and somewhat surprisingly, the covert identification capacity can be achieved without any shared key between the sender and receivers. The achievability proof relies on a random coding argument with pulse-position modulation (PPM), coupled with a second stage which performs code refinements. The converse proof relies on an expurgation argument as well as results for channel resolvability with stringent input constraints.
\end{abstract}


\flushbottom
\section{Introduction}

In contrast to Shannon's classical channel coding problem~\cite{shannon1948mathematical} (also known as the \emph{transmission problem}) in which a sender wishes to reliably send a message to a receiver through a noisy channel $W$, the problem of \emph{identification via channels}~\cite{ahlswede1989identification} (or simply the \emph{identification problem}) is rather different. It focuses on a different setting wherein a sender wishes to send an \emph{identification (ID) message} $m \in \M$ via a noisy channel $W$ to a \emph{set} of receivers $\{\R_{m'}\}_{m' \in \M}$, each observing the (same) outputs of the channel, such that every receiver $\R_{m'}$ only cares about its dedicated message $m'$ and should be able to reliably answer the following question: \emph{Is the ID message sent by the sender $m'$}? Specifically, if the ID message sent by the sender is $m$, 
\begin{itemize}
	\item The receiver $\R_{m'}$ should answer ``YES'' with high probability if $m' = m$;
	\item The receiver $\R_{m'}$ should answer ``NO'' with high probability if $m' \ne m$.
\end{itemize}
It is well known that in the transmission problem, one can reliably transmit a message of size $\exp(\Theta(n))$ over $n$ channel uses, and the pre-constant is characterized by the celebrated \emph{channel capacity} $\mathsf{C}_W \triangleq \max_{P} I(P,W)$, i.e., the mutual information of the input and output of the channel $W$ maximized over the input distribution $P$. In the identification problem, Ahlswede and Dueck~\cite{ahlswede1989identification} showed that the size of the ID message can be as large as $\exp(\exp(\Theta(n)))$, i.e.,  doubly-exponentially large in the blocklength $n$. Somewhat surprisingly, the exact pre-constant  in the $\Theta(\cdot)$ notation, which is referred to as the \emph{identification capacity}, is again $\mathsf{C}_W$~\cite{ahlswede1989identification, han1992new}. That is, the identification capacity exactly equals the channel capacity. 

Apart from reliability guarantees, recent years have witnessed an increasing attention to security concerns, especially  in networked communication systems such as the Internet of Things. From an information-theoretic perspective, the security of the classical transmission problem has been extensively studied since Wyner's seminal paper~\cite{wyner1975wire} on the wiretap channel (see~\cite{liang2009information,bloch2011physical} for surveys), and the secure identification problem has been investigated as well~\cite{ahlswede2006transmission,boche2018secure,boche2018secure2}. While most security problems are concerned with hiding the \emph{content} of information, in certain  scenarios merely the \emph{fact} that communication takes place could lead to serious consequences---thus, the sender is required to hide  the fact that he or she is communicating when he or she does so. Said differently, the sender needs to communicate {\em covertly} with respect to the warden  who is surreptitiously monitoring the communication.  This motivates the recent studies of the \emph{covert communication} problem. Following the pioneering work by Bash \emph{et al.}~\cite{bash2013limits} which demonstrates a \emph{square-root law} (SRL) (i.e., one can only transmit $\Theta(\sqrt{n})$ bits over $n$ channel uses) for covert communication, subsequent works have built on the initial work~\cite{bash2013limits} to establish information-theoretic limits for covert communication  over binary symmetric channels~\cite{CheBJ:13, che_reliable_2014-2, che2014reliable}, discrete memoryless channels (DMCs) and Gaussian channels~\cite{bloch2016covert, wang2016fundamental, tahmasbi2018first,tahmasbi2017error}, multiple-access channels~\cite{arumugam2019covert}, broadcast channels~\cite{arumugam2019embedding, tan2018time, kibloff2019embedding}, compound channels~\cite{ahmadipour2019covert}, channel with states~\cite{lee2018covert, zivarifard2019keyless}, adversarial noise channels~\cite{zhang2018covert}, relay channels~\cite{arumugam2018covert2}, \emph{etc}. In the literature, the covertness constraint requires that at the warden's side, the output distribution when communication takes place is almost indistinguishable from the output distribution when no communication takes place, and the discrepancy between the two distributions is  usually measured by the
\emph{Kullback-Leibler (KL) divergence} or the \emph{variational distance}.

In addition to covert communication which focuses on the transmission problem, there are also scenarios in which the sender wishes to reliably send an ID message to a {\em set} of receivers, and simultaneously to remain covert with respect to the warden. For instance, the government would like to inform a close contact of a confirmed COVID-19 patient to get tested. As personal contact information is unknown, the government has to send a common message over a \emph{public channel} to all the citizens. Each of the citizens would like to know whether the received  message corresponds to his or her specific index; if so, he/she must go to the hospital to be tested, otherwise nothing needs to be done. Meanwhile, as there has been no local confirmed case for a long time, the government also wishes to hide the fact that the public channel is being used from other countries (such that it is not known to the outside world that there is a new confirmed case), in order to minimize the adverse effects to economy. Here, using a covert identification scheme is superior to using a covert transmission scheme as the number of uses of the public channel (which is a valuable resource) can be significantly reduced. This motivates the study of the \emph{covert identification problem}.  Given the similarities and differences between the transmission and identification problems without covertness constraints, it is then natural to ask the following questions: (i) \emph{What is the maximum size of the  ID message with covertness constraints}, (ii) \emph{Whether the \emph{covert capacity} characterized in~\cite{bloch2016covert, wang2016fundamental} plays a role in the fundamental limits of the covert identification problem}, and (iii) \emph{Is a shared key required to ensure that the identification can take place reliably?}  

 These questions precisely set the stage of this work, and our main contributions can be summarized as follows.
\begin{itemize}
	\item Analogous to the SRL in the covert communication literature, a different form of the SRL is discovered in the covert identification problem. That is, one can send an ID message of size up to $\exp(\exp(\Theta(\sqrt{n})))$ reliably and covertly, in contrast to the standard identification problem wherein the scaling is $\exp(\exp(\Theta(n)))$.
	\item We then characterize the maximal pre-constant of the $\Theta(\cdot)$ notation in $\exp(\exp(\Theta(\sqrt{n})))$, which is referred to as the \emph{covert identification capacity}. We do so by establishing matching achievability and converse results. It turns out that the covert identification capacity equals the covert capacity; however, a key difference is that the former is achieved \emph{without any shared key} between the sender and receivers---this is in stark contrast to standard covert communication wherein the shared key is necessary~\cite{bloch2016covert} for achieving covert capacity in some regimes of the channel between the sender and receiver and the channel between the sender and the warden.  
\end{itemize}

From the achievability's perspective, the requirement of a \emph{keyless} identification code prevents us from adopting the simplest and most classical construction of identification codes proposed by Ahlswede and Dueck~\cite{ahlswede1989identification}, which relies on the existence of a capacity-achieving transmission  code for the transmission problem. This is because there does not exist a keyless covert-capacity-achieving transmission code for covert communication in general~\cite{bloch2016covert}. Instead, we show the existence of keyless capacity-achieving covert identification codes by using a \emph{random coding argument} with pulse-position
modulation (PPM) and a modified information density decoder. Some remarks about our achievability scheme are given as follows.
\begin{itemize}
    \item As is well known for the identification problem, the use of \emph{stochastic encoders} is necessary to reliably send an ID message that is super-exponentially large. Stochastic encoding means that each message $m$ is \emph{stochastically} mapped to a random sequence according to a certain probability distribution. This probability distribution is referred to as the \emph{codeword} associated with $m$, and the sequences with non-zero probability mass are referred to as the \emph{constituent sequences} of the corresponding codeword.

    \item PPM codes can be viewed as a special sub-class of constant composition codes.  In this work, we use the so-called \emph{PPM input distribution} (defined in~\eqref{eq:PPM} in Section~\ref{sec:ppm}) to \emph{directly} generate the constituent sequences of each codeword in the identification code, as the PPM input distribution has been shown to be optimal for achieving covertness~\cite{bloch2017optimal}.  On the other hand, we note that PPM has also been used to construct optimal identification codes for the (non-covert) identification problem~\cite{verdu1993explicit}. However, the way that PPM is used in~\cite{verdu1993explicit} is quite different from our application of PPM codes for the covert identification problem. Roughly speaking, the identification code of Verd\'u and Wei~\cite{verdu1993explicit} is constructed based on a capacity-achieving transmission code, i.e., the constituent sequences of each codeword are \emph{sampled} from the transmission code. They then use a \emph{binary constant-weight concatenated code}, which consists of a layer of PPM codes and two layers of Reed-Solomon codes, to provide an explicit method for the sampling process (from the transmission code), such that the resultant identification code satisfies desired requirements on the error probabilities.
    
    \item It is also worth pointing out that the random coding argument does {\em not} directly ensure the existence of a good identification code with vanishing \emph{maximum} error probabilities, due to the large message size which is of order $\exp(\exp(\Theta(\sqrt{n})))$; this issue is resolved by a careful code refinement process, which is explained in Section~\ref{sec:refinement}. Our code refinement process is different from the conventional \emph{expurgation argument} that is ubiquitous in the information theory literature, in the sense that our refinement procedure preserves the channel output distribution induced by the original code; this is critical for ensuring that the covertness constraint is satisfied.
\end{itemize}
   
The proof of the converse part for the covert identification problem is also non-standard. Roughly speaking, the converse for channel identification usually relies on the achievability for \emph{channel resolvability} for \emph{general} input distributions, as discovered in Han and Verd\'u's seminal work~\cite{han1992new}. However, such general results have not been established under stringent input constraints imposed by the covertness constraint. Instead, we circumvent this difficulty by expurgating a large number of codewords and appropriately modifying the remaining codewords of the original covert identification code, such that the resultant code satisfies certain input constraints. One can then apply the idea of~\cite{han1992new} to the new code to obtain the desired converse result. It is also worth noting that the expurgation argument used in this work differs from some relevant works on covert communication~\cite{tahmasbi2018first,zhang2018covert, zhang2019undetectable}, since the identification problem relies critically on the use of stochastic encoders (as detailed in Section~\ref{sec:model}).  

The rest of this paper is organized as follows. We provide some notational conventions and an important technical lemma in Section~\ref{sec:pre}. In Section~\ref{sec:model}, we formally introduce the covert identification problem, present the main results, and sketch the proof ideas. Sections~\ref{sec:achievability} and~\ref{sec:converse} respectively provide the detailed proofs of the achievability and converse parts for the main results. In Section~\ref{sec:conclusion}, we conclude this work and propose several promising directions for future work. 

\section{Preliminaries} \label{sec:pre}
For non-negative integers $a,b \in \mathbb{N}$, we use $[a:b]$ to denote the set of integers $\{a, a+1, \ldots, b \}$.
Random variables and their realizations are respectively denoted by uppercase and lowercase letters, e.g., $X$ and $x$. Sets are denoted by calligraphic letters, e.g., $\mathcal{X}$. Vectors of length $n$ are denoted by boldface letters, e.g., $\X$ or $\x$, while vectors of shorter length (which should be clear from the context) are denoted by underlined boldface letters, e.g., $\underline{\X}$ or $\underline{\x}$. We use $X_i$ or $x_i$ to denote the $i$-th element of a vector, and $X_a^b$ or $x_a^b$ to denote the vector $(X_a, X_{a+1}, \ldots, X_b )$ or $(x_a, x_{a+1}, \ldots, x_b)$.

Throughout this paper, logarithms $\log$ and exponentials $\exp$ are based $e$.
For two probability distributions $P$ and $Q$ over the same finite set $\mathcal{X}$, we respectively define their {\em KL-divergence}, {\em variational distance}, and {\em $\chi_2$-distance} as 
\begin{align*}
	\DD(P \Vert Q) &\triangleq \sum_{x \in \mathcal{X}} P(x) \log\frac{P(x)}{Q(x)}, \\
	\V(P, Q) &\triangleq \frac{1}{2}\sum_{x \in \mathcal{X}} |P(x) - Q(x)|, \\
	\chi_2(P \Vert Q) &\triangleq \sum_{x \in \mathcal{X}} \frac{(P(x)-Q(x))^2}{Q(x)}.
\end{align*}
We say $P$ is \emph{absolutely continuous} with respect to $Q$ (denoted by $P \ll Q$) if the support of $P$ is a subset of the support of $Q$ (i.e., for all $x \in \mathcal{X}$, $P(x) = 0$ if $Q(x) = 0$). 

Moreover, we introduce a concentration inequality that is widely used in this work---Hoeffding's inequality.
\begin{lemma}[Hoeffding's inequality~\cite{hoeffding1994probability}] \label{lemma:hoeffding}
	Suppose $\{X_i\}_{i=1}^n$ is a set of independent random variables such that $a_i \le X_i \le b_i$ almost surely, and let $X \triangleq \sum_{i=1}^n X_i$. For any $v > 0$,
	\begin{align}
		\PP\left(|X - \E(X)| \ge v \right) \le \exp\left(-\frac{2v^2}{\sum_{i=1}^n (b_i - a_i)^2}\right). \notag
	\end{align} 
\end{lemma}

\section{Problem Setting, Main Results, and Proof Sketches} \label{sec:model}

The channel between the sender and receivers is a binary-input discrete memoryless channel (BDMC) $(\mathcal{X}, \wyx,\mathcal{Y})$, and the channel between the sender and warden is another independent BDMC $(\mathcal{X}, \wzx,\mathcal{Z})$. It is assumed that $\mathcal{Y}$ and $\mathcal{Z}$ are finite alphabets, and $\mathcal{X} = \{0,1\}$ with `$0$' being the \emph{innocent symbol} and `$1$' being the symbol that carries information.\footnote{It is also possible to consider a more general setting with multiple non-zero input symbols (by following the lead of~\cite{wang2016fundamental}); however, for simplicity and ease of presentation, we focus on the binary-input setting in this work.} The channel transition probability corresponding to $n$ channel uses are denoted by $\wyxn (\y|\x) \triangleq \prod_{i=1}^n \wyx(y_i|x_i)$ and $\wzxn (\z|\x)\triangleq \prod_{i=1}^n \wzx(z_i|x_i)$. Moreover, we define
\begin{align*}
&P_0 \triangleq W_{Y|X=0}, \quad P_1 \triangleq W_{Y|X=1}, \\
&Q_0 \triangleq W_{Z|X=0}, \quad Q_1 \triangleq W_{Z|X=1}.
\end{align*}
As is common in the covert communication literature, it is assumed that (i) $Q_0 \ne Q_1$, (ii) $Q_1$ is absolutely continuous with respect to $Q_0$ (i.e., $Q_1 \ll Q_0$), and (iii) $P_1$ is absolutely continuous with respect to $P_0$ (i.e., $P_1 \ll P_0$). The first two assumptions preclude the scenarios in which covertness is always guaranteed or would never be guaranteed, while the last assumption precludes the possibility that the receivers enjoy an unfair advantage over the warden (as detailed in~\cite[Appendix~G]{bloch2016covert}). Let $\mu_0 \triangleq \min_{z: Q_0(z)>0}Q_0(z)$, $\mu_1 \triangleq \min_{z: Q_1(z)>0}Q_1(z)$, and $\widetilde{\mu} \triangleq \min\{\mu_0, \mu_1\}$

\begin{definition}[Identification codes]
	An identification code $\C$ with message set $\M$ is a collection of codewords $\{\Um\}_{m \in \M}$ and demapping regions $\{\D_m\}_{m \in \M}$, where $\Um \in \mathcal{P}(\mathcal{X}^n)$ and $\D_m \subseteq \mathcal{Y}^n$. 
\end{definition}
\begin{remark}
	In contrast to most communication problems wherein each message $m$ is deterministically mapped to a fixed sequence (the codeword) $\x \in \mathcal{X}^n$, the identification problem uses stochastic encoders such that   message $m$ is stochastically mapped to a  random sequence $\X$ according to the probability distribution $U_m \in \mathcal{P}(\mathcal{X}^n)$. Moreover, the demapping regions $\{\D_m\}_{m \in \M}$ in the identification problem are not necessarily disjoint. The use of stochastic encoders and the fact that the demapping regions are not disjoint are critical for communicating $\omega(n)$ bits of message over $n$ channel uses. With a slight abuse of terminology, we refer to the distribution $U_m$ as the \emph{codeword} for the message $m$.    
\end{remark}

The \emph{transmission status} of the sender is denoted by $T \in \{0,1\}$. Communication takes place if $T = 1$, while no communication takes place if $T = 0$.
When $T = 1$, the sender selects a message $m$ uniformly at random from $\M$. The encoder then chooses a length-$n$ sequence $\X \in \mathcal{X}^n$ according to the distribution $\Um$. When $T = 0$, the channel input is the length-$n$ zero sequence $\mathbf{0}$. For the receiver $\mathsf{R}_{m'}$ ($m' \in \M$), upon receiving the channel output $\Y \in \mathcal{Y}^n$ through the BDMC $\wyxn$, it declares that the message sent by the sender is $m'$ if and only if $\Y \in \D_{m'}$. 

The standard identification problem usually focuses on two types of error---\emph{the error probability of the first kind} which corresponds to the probability that the true message is not identified by its designated receiver, and \emph{the error probability of the second kind} which corresponds to the probability that the message is wrongly identified by some other receiver. For the covert identification problem, we introduce one more type of error---\emph{the error probability of the third kind} which corresponds to the probability that the length-$n$ zero sequence $\mathbf{0}$ (when no communication takes place) is wrongly identified as a certain message by any receiver. We formalize these notions in the following definition.

\begin{definition}[Error probabilities]
	When $T = 1$ and $m \in \M$ is sent, the \emph{error probability of the first kind} is defined as 
	\begin{align}
	P_{\mathrm{err}}^{(1)}(m) \triangleq \sum_{\x \in \mathcal{X}^n} \Uxm\wyxn(\D_m^c|\x). \notag
	\end{align}
	When $T=1$ and $m' \in \M$ is sent, the \emph{error probability of the second kind} corresponding to the receiver $\mathsf{R}_{m}$ is defined as
	\begin{align}
	P_{\mathrm{err}}^{(2)}(m,m') \triangleq \sum_{\x \in \mathcal{X}^n} \Uxhm\wyxn(\D_m|\x). \notag
	\end{align}
	When $T=0$ and the length-$n$ zero sequence is sent through the channel, the \emph{error probability of the third kind}  corresponding to the receiver $\mathsf{R}_{m}$ is defined as
	\begin{align}
	P_{\mathrm{err}}^{(3)}(m) \triangleq P_0^{\otimes n}(\D_m). \notag
	\end{align}
	Furthermore, let the corresponding \emph{maximum error probabilities} (over all the messages or all pairs of distinct messages) be
	\begin{align}
	&P_{\mathrm{err}}^{(1)} \triangleq \max_{m \in \M} P_{\mathrm{err}}^{(1)}(m), \notag \\
	& P_{\mathrm{err}}^{(2)} \triangleq \max_{(m,m') \in \M^2: m \ne m'} P_{\mathrm{err}}^{(2)}(m,m'),
	\notag \\ 
	& P_{\mathrm{err}}^{(3)} \triangleq \max_{m \in \M} P_{\mathrm{err}}^{(3)}(m). \notag
	\end{align} 
\end{definition}  

Let $\widehat{Q}^n_{\C}(\z)$ be the output distribution on $\mathcal{Z}^n$ for the warden induced by the identification code, which takes the form 
\begin{align}
\widehat{Q}^n_{\C}(\z) \triangleq \frac{1}{|\M|}\sum_{m \in \M} \sum_{\x \in \mathcal{X}^n} \Uxm \wzxn(\z|\x), \ \forall \z \in \mathcal{Z}^n. \label{eq:qhat}
\end{align}
We adopt the widely-used KL-divergence metric $\DD(\widehat{Q}^n_{\C} \Vert Q_0^{\otimes n})$ to measure covertness with respect to the warden. 
\begin{definition}[Covertness] \label{def:covert}
	The communication is \emph{$\delta$-covert} if the KL-divergence between the distribution $\widehat{Q}^n_{\C}$ (when $T=1$) and $Q_0^{\otimes n}$ (when $T=0$) is bounded from above by $\delta$, i.e., $$\DD(\widehat{Q}^n_{\C} \Vert Q_0^{\otimes n}) \le \delta.$$
\end{definition}

Let $\pi_{1|0}$ and $\pi_{0|1}$ respectively  be the probabilities of false alarm (i.e., making an error when $T=0$) and missed detection (i.e., making an error when $T=1$) of the warden's hypothesis test. By using the definition of the variational distance and  Pinsker's inequality,  we see that the optimal test satisfies 
\begin{align*}
	\pi_{1|0} + \pi_{0|1} = 1 - \V(\widehat{Q}^n_{\C}, Q_0^{\otimes n}) \ge 1 -  \sqrt{\DD(\widehat{Q}^n_{\C} \Vert Q_0^{\otimes n})}.
\end{align*}  
Thus, a small $\DD(\widehat{Q}^n_{\C} \Vert Q_0^{\otimes n})$ implies a large sum-error $\pi_{1|0} + \pi_{0|1}$. This provides an  operational meaning of the covertness metric in Definition~\ref{def:covert}. As discussed in prior works such as~\cite{hou2014effective, wang2016fundamental,tahmasbi2018first}, the variational distance metric $\V(\widehat{Q}^n_{\C}, Q_0^{\otimes n})$ is perhaps a better metric under the specific assumption that $T=0$ and $T=1$ occur with equal probabilities, since it directly connects  to the average error probability of detection; however, the above assumption does not hold in general, thus both KL-divergence and variational distance are deemed to be appropriate metrics in the literature.      

\begin{definition} \label{def:rate}
	A rate $R$ is said to be \emph{$\delta$-achievable} if there exists a sequence of identification codes with increasing blocklength $n$ such that  
	\begin{align}
	&\liminf_{n \to \infty} \frac{\log\log |\M|}{\sqrt{n}} \ge R, \quad  \DD\left(\widehat{Q}^n_{\C} \Vert Q_0^{\otimes n}\right) \le \delta, \notag \\
	&\lim_{n \to \infty} P_{\mathrm{err}}^{(1)} = \lim_{n \to \infty} P_{\mathrm{err}}^{(2)}= \lim_{n \to \infty} P_{\mathrm{err}}^{(3)} = 0. \notag
	\end{align}
	The \emph{$\delta$-covert identification capacity} $C_{\delta}$ is defined as the supremum of all $\delta$-achievable rates.
\end{definition}

Note that the coding rate $R$ in the covert identification problem is defined as the {\em iterated} logarithm of the size of the message set $|\M|$ normalized by $\sqrt{n}$, which implies that the message size (if $R>0$) is of order $\exp(\exp(\Theta(\sqrt{n})))$. This intuitively makes sense because the channel identification problem usually allows the message size to be as large as $\exp(\exp(\Theta(n)))$, but the stringent covertness constraint reduces the exponent from $\Theta(n)$ to $\Theta(\sqrt{n})$. In the following, we present the main result that characterizes the $\delta$-covert identification capacity of BDMCs.

\subsection*{Main result: The covert identification capacity}

\begin{theorem} \label{thm}
	For any BDMCs $\wyx$ and $\wzx$ satisfying $Q_1 \ne Q_0$, $Q_1 \ll Q_0$, and $P_1 \ll P_0$, the $\delta$-covert identification capacity is given by
	\begin{align}
	C_{\delta} = \sqrt{\frac{2\delta}{\chi_2(Q_1\Vert Q_0)}} \DD(P_1 \Vert P_0). \notag
	\end{align} 
\end{theorem}

Some remarks are in order.
\begin{enumerate}
	\item Analogous to the canonical covert communication problem, we notice that the SRL also holds for the covert identification problem albeit with message size $\exp(\exp(\Theta(\sqrt{n})))$. Furthermore, the $\delta$-covert identification capacity is exactly the same as the $\delta$-covert capacity derived in~\cite{wang2016fundamental,bloch2016covert}.  
	\item In stark contrast to  the standard covert communication problem~\cite{bloch2016covert} in which a shared key is needed to achieve the covert capacity when the channels $\wyx$ and $\wzx$ satisfy $\DD(P_1 \Vert P_0) \le \DD(Q_1 \Vert Q_0)$, Theorem~\ref{thm} above shows that regardless of the values of $\DD(P_1 \Vert P_0)$ and $\DD(Q_1 \Vert Q_0)$, the $\delta$-covert identification capacity is always achievable {\em without any shared key}. Intuitively, this is because the message size in our setting scales as $\exp(\omega(n))$, which automatically allows us to satisfy the requirements on the shared key via proof techniques from channel resolvability~\cite{han1993approximation} since it is well known that an exponential message size (of a suitably large exponent) suffices to drive the approximation error (of the target and synthesized distributions) to zero. This is reflected in Lemma~\ref{lemma:covert2} in our achievability proof. 
\end{enumerate}
The achievability and converse parts of Theorem~\ref{thm} are respectively proved in Section~\ref{sec:achievability} and~\ref{sec:converse}. Before going into details, we first provide proof sketches of both parts in the following. 

\subsection*{High-level intuitions and proof sketch of achievability}
It has been well understood from the covert communication literature that to ensure covertness with respect to a warden, the average \emph{Hamming weight} of the length-$n$ channel inputs should be at most $\Theta(\sqrt{n})$, and the exact pre-constant has also been characterized as a function of the channel $\wzx$ and covertness parameter $\delta$. In recent years various coding schemes with theoretical guarantees on covertness have been developed, among which the two most widely adopted schemes are perhaps the \emph{low-weight i.i.d. random codes}~\cite{bloch2016covert,wang2016fundamental,CheBJ:13} and \emph{PPM random codes}~\cite{bloch2017optimal,tahmasbi2018first,tahmasbi2017error}. The former generates each bit of the length-$n$ sequence independently according to the Bernoulli distribution $\mathrm{Bern}(\Theta(1/\sqrt{n}))$, while the latter is a more structured approach in which the random code is generated according to the so-called PPM input distribution $\ppmx$ (formally defined in~\eqref{eq:PPM}) such that certain predefined intervals of length $\Theta(\sqrt{n})$ contain only a single one. Although these two approaches are equally favorable from the perspective of covertness, it is preferable to use PPM random codes from the perspective of identification, as the property of \emph{constant Hamming weights} helps to circumvent various challenges in the analysis of error probabilities; this is discussed in detail in Remark~\ref{remark:compare}, Section~\ref{sec:achievability}. Thus, our achievability scheme is based on the PPM random codes.

It is also worth noting that merely requiring the average Hamming weight to be low is not sufficient for achieving covertness. The second requirement is that the target output distribution $\widehat{Q}^n_{\C}$ (induced by the code) should be close to the synthesized output distribution $\ppmz$ (induced by the input distribution $\ppmx$ and the channel $\wzxn$). That is, the KL divergence of these two distributions should tend to zero. From channel resolvability~\cite{han1993approximation} we know that if a code contains at least $\exp\{I(\ppmx,\wzxn)+ \omega(n^{1/4})\}$ length-$n$ sequences and each sequence is generated according to $\ppmx$, the KL divergence $\DD(\widehat{Q}^n_{\C} \Vert \ppmz)$ will tend to zero with high probability. In the identification problem, as the message size is $\exp(\omega(n))$, using the PPM random codes will automatically fulfill the second requirement, thus covertness can be achieved.

We now turn to discuss how to bound the error probabilities for the identification problem when the size of the ID message $|\mathcal{M}|$ scales as $\exp(e^{\Theta(\sqrt{n})})$. Our scheme is as follows. For each message $m \in \mathcal{M}$, we set the codeword $U_m$ to be a uniform distribution over $N$ constituent sequences $\{\x_{m,i}\}_{i=1}^N$ that are generated independently according to the PPM input distribution $\ppmx$, and the demapping region $\D_m$ to be the \emph{union} of $N$ \emph{conditional typical sets} $\{\mathcal{F}_{\x_{m,i}}\}_{i=1}^N$ on $\mathcal{Y}^n$ (where $\mathcal{F}_{\x_{m,i}} \subseteq \mathcal{Y}^n$ is the conditional typical set with respect to $\wyx$ and a fixed $\x_{m,i}$, which is formally defined in Section~\ref{sec:stage1_1}).  
\begin{itemize}
	\item \underline{Error probability of the first kind $P_{\mathrm{err}}^{(1)}$:} Now, suppose the sender wishes to send the ID message $m \in \mathcal{M}$ (thus each of the $N$ constituent sequences $\{\x_{m,i} \}_{i=1}^N$ will be sent with equal probability), and consider the receiver $\mathsf{R}_{m}$ with demapping region $\D_m$. Intuitively, regardless of which constituent sequence $\x_{m,i}$ is sent, the channel outputs $\y$ will fall into the corresponding conditional typical set $\mathcal{F}_{\x_{m,i}}$ with high probability. This also implies that $\y$ will belong to the demapping region $\D_m$ with high probability, or equivalently $\wyxn(\D_m^c|\x_{m,i}) \le \varepsilon_n$ for some $\varepsilon_n \to 0$, since $\D_m$ is a superset of  $\mathcal{F}_{\x_{m,i}}$ by definition.  One can then treat each $\wyxn(\D_m^c|\x_{m,i})$ as a random variable with expectation at most $\varepsilon_n$. By applying Hoeffding's inequality to the $N$ i.i.d. random variables $\{\wyxn(\D_m^c|\x_{m,i})\}_{i=1}^N$, one can show that their empirical mean (which turns out to be $P_{\mathrm{err}}^{(1)}(m)$) tends to zero with probability approximately $1 - \exp(-\Theta(N))$. Moreover, one needs to take a union bound over all the messages to ensure that the \emph{maximum} error probability $P_{\mathrm{err}}^{(1)} = \max_{m \in \mathcal{M}}P_{\mathrm{err}}^{(1)}(m)$ vanishes with high probability; this requires $|\mathcal{M}|\cdot\exp(-\Theta(N))$ to approach zero.
	
	\item \underline{Error probability of the second kind $P_{\mathrm{err}}^{(2)}$:} Consider the receiver $\mathsf{R}_{m}$, and suppose the sender sends another ID message $m'$, where $m' \ne m$. We first \emph{fix} the $N$ constituent sequences $\{\x_{m,j} \}_{j=1}^N$ and the demapping region $\D_{m}$ for $\mathsf{R}_{m}$. Our goal is to upper bound the probability that message $m'$ is wrongly identified by $\mathsf{R}_{m}$, i.e., $\frac{1}{N}\sum_{i \in [1:N]}\wyxn(\D_{m}|\x_{m',i})$. Instead of calculating each term $\wyxn(\D_{m}|\x_{m',i})$ directly, we follow Shannon's approach to calculate the average error probability $\E_{\ppmx}(\wyxn(\D_{m}|\X_{m',i}))$, in which the average is done over the generation of the constituent sequence $\X_{m',i} \sim \ppmx$. It turns out that as long as $N<\exp\{C_{\delta}\sqrt{n}\}$  (thus the demapping region $\D_{m}$, which consists of $N$ conditional typical sets, is not too large), this average error probability will tend to zero. By applying Hoeffding's inequality to the $N$ i.i.d. random variables $\{\wyxn(\D_{m}|\X_{m',i}) \}_{i=1}^N$, one can show that their empirical mean tends to zero with probability approximately $1 - \exp(-\Theta(N))$. Additionally, by taking the randomness of the constituent sequences $\{\X_{m,j}\}_{j=1}^N$ for $m$  into account and letting $\mathbf{D}_{m}$ be the \emph{chance variable} corresponding to $\D_m$, one can also show that the empirical mean of $\{\wyxn(\mathbf{D}_{m}|\X_{m',i}) \}_{i=1}^N$ tends to zero with probability at least $1 - \exp(-\Theta(N))$. It then remains to take a union bound over all the message pairs $(m,m')$ to ensure that the \emph{maximum} error probability $P_{\mathrm{err}}^{(2)}$ vanishes with high probability over the code generation; this requires $|\mathcal{M}|^2\cdot\exp(-\Theta(N))$ to approach zero.   
\end{itemize}

The two requirements $N<\exp\{C_{\delta}\sqrt{n}\}$ and $|\mathcal{M}|^2\cdot\exp(-\Theta(N)) \to 0$ motivate us to set $N = \exp\{(1-(\eta/2))C_{\delta}\sqrt{n}\}$ and $|\mathcal{M}| = \exp\{e^{(1-\eta)C_{\delta}\sqrt{n}} \}$, where $\eta \in (0,1)$ can be made arbitrarily small. Note that $\lim_{\eta \to 0^+}\liminf_{n \to \infty}\log\log|\mathcal{M}|/\sqrt{n} = C_{\delta}$.

Finally, we show that with high probability over the code generation, the \emph{average} error probability of the third kind $|\M|^{-1} \sum_{m \in \M} P_{\mathrm{err}}^{(3)}(m)$ tends to zero, and we further apply a code refinement process (see Section~\ref{sec:refinement}) to ensure that the {\em maximum} error probability $\max_{m \in \M} P_{\mathrm{err}}^{(3)}(m)$ also vanishes.

\subsection*{Proof sketch of converse}
First, one can show that for any identification code $\C$ satisfying the covertness constraint,  the expected Hamming weight (over the message and \emph{all} the codewords) of the constituent sequences in $\C$ is at most $d\sqrt{n}$ for some constant $d > 0$ that can be explicitly characterized (as shown in Lemma~\ref{lemma:constraint}). Based on the original code $\C$, one can use an expurgation argument to construct another code $\C'$ in which \emph{all} the constituent sequences have Hamming weight at most  $d\sqrt{n}(1+\varepsilon_n)$ for some $\varepsilon_n \to 0$, and it  simultaneously ensures that (i) the message size in $\C'$ is almost as large as that in $\C$, and (ii) the error probabilities of $\C'$ are almost as small as those in $\C$ (as shown in Lemma~\ref{lemma:expurgation}). Due to these favourable properties, it can be shown that the upper bound on the message size of the original code $\C$ is essentially almost the same as that of the new code $\C'$. Thus, it suffices to analyze the new code $\C'$, and the fact that all the constituent sequences in $\C'$ have Hamming weight at most $d\sqrt{n}(1+\varepsilon_n)$ is important for the subsequent analysis.

The rest of the proof relies critically on the achievability for channel resolvability for \emph{general low-weight} input distributions. Roughly speaking, there exists a $K \in \mathbb{N}^+$ such that each of the codeword in $\C'$ can be approximated by a \emph{$K$-type distribution} (formally defined in Definition~\ref{def:type}) with a vanishing approximation error. Meanwhile, due to the fact that $\C'$ has small error probabilities, each pair of codewords should be sufficiently well-separated from each other. Thus, every pair of $K$-type distributions (which are used to approximate codewords) should also be well-separated from each other. By noting that the number of \emph{distinct} $K$-type distributions on $\{0,1\}^n$ is at most $2^{nK}$, one can in turn deduce that the number of codewords (or equivalently, the message size) in $\C'$ is at most $2^{nK}$. Finally, choosing an appropriate value of $K$ leads to the desired converse result.

\section{Achievability} \label{sec:achievability}

The achievability proof is partitioned into two stages. In the first stage, we use a random coding argument to show the existence of a ``weak'' covert identification code---we say a code to be a ``weak'' covert identification code if it ensures (i) covertness, (ii) vanishing maximum error probabilities of the first and second kinds $P_{\mathrm{err}}^{(1)}$ and $P_{\mathrm{err}}^{(2)}$, and (iii) a vanishing \emph{average} (rather than maximum) error probability of the third kind $(1/|\M|)\sum_{m \in \M} P_{\mathrm{err}}^{(3)}(m)$. In the second stage, we apply a code refinement process to the ``weak'' covert identification code, such that the refined code satisfies all the criteria for the three error probabilities and covertness in Definition~\ref{def:rate}.
 
We first provide a detailed introduction of PPM in Subsection~\ref{sec:ppm}. The first stage of the achievability is described in Subsection~\ref{sec:stage1_1} and proved in Subsection~\ref{sec:proof_lemma1}, while the second stage is presented in Subsection~\ref{sec:refinement}. Finally, in Subsection~\ref{sec:other}, we discuss the reasons why other identification schemes are not applicable to the covert identification problem. Table~\ref{table:1} summarizes all the parameters that are used in this section. 

\subsection{Pulse-Position Modulation (PPM)} \label{sec:ppm}
Let $$l \triangleq\left \lfloor \sqrt{ \frac{(2\delta-n^{-1/3})n}{\chi_2(Q_1\Vert Q_0)} }  \right\rfloor$$ be the {\em weight parameter}, and $(w,s)$ be non-negative integers such that $w \triangleq \left \lfloor{n/l}\right \rfloor$ and $s \triangleq n - wl$.
We use $\ux \in \mathcal{X}^w, \uy \in \mathcal{Y}^w, \uz \in \mathcal{Z}^w$ to denote vectors of length $w$. We also let $ \text{wt}_{\mathrm{H}}(\ux)$ denote the number of ones, or the {\em Hamming weight}, of the vector  $\ux$. Let 
\begin{align}
\pwx(\ux) \triangleq \begin{cases}
1/w, & \text{if} \ \text{wt}_{\mathrm{H}}(\ux) = 1, \\
0, &\text{otherwise},
\end{cases} \notag
\end{align}
be the distribution on $\mathcal{X}^w$ such that $\pwx(\ux)$ is non-zero if and only if $\ux$ has Hamming weight one. The corresponding output distributions $\pwy$ and $\pwz$ are respectively given by 
\begin{align}
&\pwy(\uy) \triangleq \sum_{\ux \in \mathcal{X}^w} \pwx(\ux) \wyxw(\uy|\ux), \quad \mbox{and} \label{eq:s1} \\
&\pwz(\uz) \triangleq \sum_{\ux \in \mathcal{X}^w} \pwx(\ux) \wyxw(\uz|\ux). \label{eq:s2}
\end{align} 
For each $i \in [1:l]$, we define $\ux^{(i)} \triangleq x_{(i-1)w+1}^{iw}$ as the length-$w$ subsequence of $\x$ that  comprises consecutive elements from $x_{(i-1)w+1}$ to $x_{iw}$. Thus, a length-$n$ vector $\x$ can be represented as $\x = [\ux^{(1)}, \ldots, \ux^{(l)}, x_{wl+1}^n]$, where $x_{wl+1}^n$ is of length $s$. The \emph{PPM input distribution} is thus defined as 
\begin{align}
\ppmx(\x) \triangleq \prod_{i=1}^l \pwx(\ux^{(i)}) \cdot \mathbbm{1}\left\{\text{wt}_{\mathrm{H}}(x_{wl+1}^n) = 0\right\}. \label{eq:PPM}
\end{align}
That is, we require each PPM-generated vector, also called a {\em PPM-sequence}, $\x$ to contain exactly $l$ ones; in particular, each of the first $l$ intervals $[1:w], [w+1:2w], \ldots, [(l-1)w+1:lw]$ contains a single one, and the last interval $[wl+1:n]$ contains all zeros. The PPM-induced output distributions $\ppmy$ on $\mathcal{Y}^n$ and $\ppmz$ on $\mathcal{Z}^n$ are respectively given by
\begin{align}
\ppmy(\y)  &\triangleq \sum_{\x\in \mathcal{X}^n} \ppmx(\x) \wyxn(\y|\x) \notag\\
&= \sum_{\x\in \mathcal{X}^n} \prod_{i=1}^l \pwx(\ux^{(i)})  \mathbbm{1}\left\{ \text{wt}_{\mathrm{H}}(x_{wl+1}^n) = 0\right\} \wyxn(\y|\x) \notag  \\
&= \left(\prod_{i=1}^l \pwy(\uy^{(i)})\right) \cdot P_0^{\otimes s}(y_{wl+1}^n), \label{eq:mean}\quad\mbox{and} \\
\ppmz(\z)  &\triangleq \sum_{\x\in \mathcal{X}^n} \ppmx(\x) \wzxn(\z|\x) = \left(\prod_{i=1}^l \pwz(\uz^{(i)})\right) \cdot Q_0^{\otimes s}(z_{wl+1}^n).\notag
\end{align}

\subsection{Existence of a ``weak'' covert identification code} \label{sec:stage1_1}

First recall that a ``weak'' covert identification code satisfies all the criteria in Definition~\ref{def:rate} except that it only ensures a vanishing average (rather than maximum) error probability of the third kind. In the following, we show such a ``weak'' covert identification code exists by using a random coding argument.

\subsubsection{Encoder and Demapping regions}

Let $\eta\in (0,1)$ be arbitrary, $t \triangleq l/\sqrt{n}$ be the normalized weight parameter, $R = (1-\eta)t\DD(P_1 \Vert P_0)$, and $R' = (1-(\eta/2))t\DD(P_1 \Vert P_0)$. The size of the message set is set to be $|\M| = \exp(e^{R\sqrt{n}})$.
For each message $m \in \M$, we generate $N \triangleq e^{R'\sqrt{n}}$ constituent sequences $\{\x_{m,i} \}_{i=1}^N$ independently according to $\ppmx$, and the codeword $\Um$ is the uniform distribution over the multiset $\{\x_{m,i} \}_{i=1}^N$, i.e.,
\begin{align}
\Uxm \triangleq \frac{1}{N}\sum_{i=1}^N \mathbbm{1}\left\{\x = \x_{m,i} \right\}, \quad \ \forall \x \in \mathcal{X}^n. \notag
\end{align} 
That is, we send each of the sequences $\{\x_{m,i} \}_{i=1}^N$ with equal probability when $m$ is the true message. 

Let $\gamma \triangleq (1-\epsilon)t\DD(P_1 \Vert P_0)$, where $0 < \epsilon < \eta/2$. To specify the demapping region $\D_m$ for each message $m \in \M$, we first define the set $\F_{\x}$ for each $\x \in \mathcal{X}^n$ as 
\begin{align}
\F_{\x} \triangleq \left\{\y \in \mathcal{Y}^n: \log \frac{\wyxn(\y|\x)}{P_0^{\otimes}(\y)} > \gamma\sqrt{n} \right\}. \notag
\end{align}
The demapping region for each $m$ is $\D_m \triangleq \cup_{i \in [1:N]} \F_{\x_{m,i}}$.

\subsubsection{Error probabilities and distributions of interest}
Based on the encoding scheme described above and the constituent sequences $\{\x_{m,i} \}_{i=1}^N$ for each $m \in \M$, the error probabilities of the first and second kinds can be rewritten as 
\begin{align}
&P_{\mathrm{err}}^{(1)}(m) = \sum_{\x \in \mathcal{X}^n} \frac{1}{N}\sum_{i=1}^N \mathbbm{1}\left\{\x = \x_{m,i} \right\} \wyxn(\D_m^c|\x) = \frac{1}{N}\sum_{i=1}^N \wyxn(\D_m^c|\x_{m,i}), \label{eq:perr1} \quad\mbox{and} \\
&P_{\mathrm{err}}^{(2)}(m,m') = \frac{1}{N} \sum_{i=1}^N \wyxn(\D_m|\x_{m',i}), \label{eq:perr2}
\end{align}
and the output distribution $\widehat{Q}^n_{\C}$ on $\mathcal{Z}^n$, which is first defined in~\eqref{eq:qhat}, can be rewritten as
\begin{align}
\widehat{Q}^n_{\C}(\z) = \frac{1}{|\M|}\sum_{m \in \M} \frac{1}{N}\sum_{i=1}^N \wzxn(\z|\x_{m,i}).\notag
\end{align}

\subsubsection{Performance guarantees} Lemma~\ref{lemma:initial} below shows that with high probability, the randomly generated identification code is a ``weak'' covert identification code, in the sense that it only has a vanishing {\em average} (and not maximum) error probability of the third kind.

\begin{lemma} \label{lemma:initial}
	There exist vanishing sequences $\kappa_n, \varepsilon_n^{(1)}, \varepsilon_n^{(2)}, \varepsilon_n^{(3)} > 0$ (depending on the channels $\wyx, \wzx$ and the covertness parameter $\delta$) such that with probability at least $1 - \kappa_n$ over the code generation process, the randomly generated code satisfies 
	\begin{align}
	&\max_{m \in \M} P_{\mathrm{err}}^{(1)}(m)  \le \varepsilon_n^{(1)}, \ \max_{(m,m') \in \M^2: m \ne m'} P_{\mathrm{err}}^{(2)}(m,m')  \le \varepsilon_n^{(2)},  \notag \\
	&\frac{1}{|\M|}\sum_{m \in \M} P_{\mathrm{err}}^{(3)}(m) \le \varepsilon_n^{(3)}, \quad \ \DD\left(\widehat{Q}^n_{\C} \Vert Q_0^{\otimes n}\right) \le \delta. \notag
	\end{align}
\end{lemma}

\begin{table}[]
	\small
	\centering
	\caption{Table of parameters used for achievability}
	\begin{tabular}{|l|l|l|l|}
		\hline
		\textbf{Symbol} & \textbf{Description} & \textbf{Equality/Range}  \\ \hline		
		$\delta$      & Covertness parameter     & $\delta > 0$  \\ \hline
		$\mathcal{M}$      & Message set     & $\mathcal{M} = [1:|\mathcal{M}|]$  \\ \hline
		$U_m$      & Codeword for message $m$     & $U_m \in \mathcal{P}(\mathcal{X}^n) $  \\ \hline
		$\D_m$      & Demapping region for message $m$     & $\D_m \subseteq \mathcal{Y}^n $  \\ \hline
		$l$   & Hamming weight of each PPM-sequence &  $l = \big\lfloor \sqrt{ \frac{(2\delta-n^{-1/3})n}{\chi_2(Q_1\Vert Q_0)} }  \big\rfloor $   \\ \hline
		$t$   & Normalized Hamming weight of each PPM-sequence &  $t = l/\sqrt{n}$   \\ \hline
		$\eta, \epsilon$   & Slackness parameters &  $\eta \in (0,1)$, $\epsilon \in (0, \eta/2)$   \\ \hline
		$R'$   & Normalized number of constituent sequences of each codeword &  $R' = (1-(\eta/2))t \DD(P_1 \Vert P_0)$   \\ \hline
		$N$   & Number of constituent sequences of each codeword &  $N = e^{R'\sqrt{n}}$   \\ \hline
		$\{\x_{m,i} \}_{i=1}^N$       & Constituent sequences for message $m$     & $\x_{m,i} \in \{0,1\}^n $  \\ \hline
		$R$   & Rate of the covert identification code ($R \triangleq (\log\log|\mathcal{M}|)/\sqrt{n}$) &  $R = (1-\eta)t \DD(P_1 \Vert P_0)$   \\ \hline
		$\gamma$   & Parameter specifying the demapping regions &  $(1-\epsilon)t\DD(P_1 \Vert P_0)$   \\ \hline
		$w$   & Length of each interval (for PPM) that contains a single one  &  $w = \lfloor n/l \rfloor $   \\ \hline
		$s$   & Length of the last interval that contains all zeros &  $s = n - wl$   \\ \hline
		$\ppmx$   & PPM input distribution &  Defined in~\eqref{eq:PPM}   \\ \hline
		$\mu$   & Slackness parameter &  $\mu \in (0,\min\{R'-R, \gamma-R' \})$   \\ \hline		
		$\xi$ & Parameter depending on the channel $\wyx$  &  $\xi = \sum_{y \in \mathcal{Y}}P_1(y)^2/P_0(y)$   \\ \hline
		$\mu_0$ & Parameter depending on the channel $\wzx$  &  $\mu_0 = \min_{z: Q_0(z)>0}Q_0(z)$   \\ \hline
		$\mu_1$ & Parameter depending on the channel $\wzx$  &  $\mu_1 = \min_{z: Q_1(z)>0}Q_1(z)$   \\ \hline
		$\widetilde{\mu}$ & Parameter depending on the channel $\wzx$ &	$\widetilde{\mu} = \min\{\mu_0, \mu_1\}$  \\ \hline
	\end{tabular}
\label{table:1}
\end{table}

\subsection{Proof of Lemma~\ref{lemma:initial}} \label{sec:proof_lemma1}

\subsubsection{Analysis of $P_{\mathrm{err}}^{(1)}$}
Consider a fixed message $m \in \M$. By recalling Eqn.~\eqref{eq:perr1} and noting that $\D_m^c \subseteq \F_{\x_{m,i}}^c$, we have 
\begin{align}
P_{\mathrm{err}}^{(1)}(m) &= \frac{1}{N}\sum_{i=1}^N  \wyxn(\D_m^c|\x_{m,i}) \le \frac{1}{N}\sum_{i=1}^N  \wyxn(\F_{\x_{m,i}}^c|\x_{m,i}). \label{eq:average}
\end{align}
Note that each $\x_{m,i}$ is generated according to $\ppmx$, and
\begin{align}
\E_{\ppmx}\left(\wyxn(\F_{\X}^c|\X)\right) &= \sum_{\x} \ppmx(\x)\sum_{\y} \wyxn(\y|\x) \mathbbm{1}\left\{\log\frac{\wyxn(\y|\x)}{P_0^{\otimes}(\y) } \le \gamma\sqrt{n} \right\} \notag\\
&=\sum_{\x} \ppmx(\x)\sum_{\y} \wyxn(\y|\x) \mathbbm{1}\left\{\sum_{j=1}^n \log\frac{\wyx(y_j|x_j)}{P_0(y_j) }\le  \gamma\sqrt{n} \right\}\notag \\
&= \sum_{\x} \ppmx(\x) \sum_{\y} \wyxn(\y|\x) \mathbbm{1} \left\{ \sum_{j: x_j = 1} \log\frac{P_1(y_j)}{P_0(y_j) } \le \gamma\sqrt{n} \right\}, \label{eq:ignore}
\end{align}
where~\eqref{eq:ignore} holds since $\log \frac{\wyx(y_j|x_j)}{P_0(y_j)} = \log \frac{P_0(y_j)}{P_0(y_j)} = 0$ for all $j$ such that $x_j = 0$. Without loss of generality, we define $\x^* \in \mathcal{X}^n$ as the weight-$l$ vector such that $x^*_{(j-1)w+1} = 1$ for $j \in [1:l]$, thus~\eqref{eq:ignore} also equals
\begin{align}
&\sum_{\y} \wyxn(\y|\x^*) \mathbbm{1}\left\{\sum_{j= 1}^{l} \log\frac{P_1(y_{(j-1)w+1})}{P_0(y_{(j-1)w+1}) } \le \gamma\sqrt{n} \right\} = \PP_{P_1^{\otimes l}}\left( \sum_{j= 1}^{l}\log\frac{P_1(Y_{(j-1)w+1})}{P_0(Y_{(j-1)w+1}) } \le \gamma\sqrt{n} \right). \notag
\end{align}
Note that the random variables $\{\log\frac{P_1(Y_{(j-1)w+1})}{P_0(Y_{(j-1)w+1})}\}_{j \in [1:l]}$ are independent and bounded,   $\E(\sum_{j=1}^{l}\log\frac{P_1(Y_{(j-1)w+1})}{P_0(Y_{(j-1)w+1})}) = l\DD(P_1 \Vert P_0)$, and $\gamma\sqrt{n} \triangleq (1-\epsilon)l\DD(P_1 \Vert P_0)$. By applying Hoeffding's inequality (Lemma~\ref{lemma:hoeffding}), we have 
\begin{align}
\PP_{P_1^{\otimes l}}\left( \sum_{j= 1}^{l}\log\frac{P_1(Y_{(j-1)w+1})}{P_0(Y_{(j-1)w+1})} \le \gamma\sqrt{n} \right) \le 2e^{-c_1\sqrt{n}}, \label{eq:acc}
\end{align}
for some constant $c_1 > 0$. 

Let $\mu$ be a constant satisfying $0 < \mu < \min\{R'-R, \gamma - R' \}$, $\beta_n \triangleq 2e^{-c_1\sqrt{n}}$, and $\alpha_n \triangleq \max\{2\beta_n, e^{-\mu\sqrt{n}/2} \}$.
Consider the $N$ i.i.d. random variables $\{\wyxn(\F_{\X_{m,i}}^c|\X_{m,i})\}_{i\in[1:N]}$ which correspond to the right-hand side (RHS) of~\eqref{eq:average}. Note that each random variable belongs to $[0,1]$, and the expectation is at most $\beta_n$ according to~\eqref{eq:acc}. By applying Hoeffding's inequality again and noting that $\alpha_n - \beta_n \ge e^{-\mu\sqrt{n}/2}/2$, we have  
\begin{align}
&\PP\left(\frac{1}{N}\sum_{i=1}^N \wyxn(\F_{\X_{m,i}}^c|\X_{m,i}) \ge \alpha_n \right) \le \exp\left\{-2N (\alpha_n - \beta_n)^2\right\} \le \exp\left\{-\frac{1}{2}e^{(R'-\mu)\sqrt{n}} \right\}.\notag
\end{align}
Therefore, a union bound over all the messages $m \in \M$ yields 
\begin{align*}
\PP\left(\max_{m \in \M} P_{\mathrm{err}}^{(1)} \ge \alpha_n \right) &= \PP\left(\exists m \in \M: \frac{1}{N}\sum_{i=1}^N  \wyxn(\D_m^c|\X_{m,i}) \ge \alpha_n \right)  \\
&\le \sum_{m \in \M} \PP\left(\frac{1}{N}\sum_{i=1}^N  \wyxn(\F_{\X_{m,i}}|\X_{m,i}) \ge \alpha_n \right) \\
&= \exp\left\{-\frac{1}{2}e^{(R'-\mu)\sqrt{n}} + e^{R\sqrt{n}} \right\},
\end{align*}
which vanishes since the choice of $\mu$ ensures $R'-\mu > R$.

\subsubsection{Analysis of $P_{\mathrm{err}}^{(2)}$} Consider a fixed message pair $(m,m') \in \M^2$ such that $m \ne m'$. Recall from Eqn.~\eqref{eq:perr2} that
\begin{align}
P_{\mathrm{err}}^{(2)}(m,m') =  \frac{1}{N} \sum_{i=1}^N \wyxn(\D_m|\x_{m',i}). \label{eq:average2}
\end{align}
Suppose the multiset of PPM-sequences $\{\x_{m,j} \}_{j=1}^N$ (i.e., $\ppmx(\x_{m,j}) \ne 0$) for message $m$ is fixed, thus the demapping region $\D_m$ is also fixed. Since $\D_m = \cup_{j \in [1:N]}\F_{\x_{m,j}}$, we have  
\begin{align}
\E_{\ppmx}\left(\wyxn(\D_m|\X)\right) \le \sum_{j=1}^N \E_{\ppmx}\left(\wyxn(\F_{\x_{m,j}}|\X)\right). \label{eq:exp}
\end{align}
\begin{lemma} \label{claim:perr2}
	Let $\xi \triangleq \sum_{y \in \mathcal{Y}}\frac{P_1(y)^2}{P_0(y)}$. For any PPM-sequence $\widetilde{\x} \in \mathcal{X}^n$, we have 
	\begin{align}
	\E_{\ppmx}\left(\wyxn(\F_{\widetilde{\x}}|\X)\right) \le \exp\left\{-\gamma\sqrt{n} + l(\xi-1)/w \right\}. \label{eq:hold}
	\end{align}
\end{lemma}
The proof of Lemma~\ref{claim:perr2} can be found in Appendix~\ref{appendix:lemma3}.
Combining~\eqref{eq:exp} and Lemma~\ref{claim:perr2}, we obtain that for any fixed demapping region $\D_m$ that corresponds to a multiset of PPM-sequences $\{\x_{m,j} \}_{j =1}^N$, the expectation of the random variable $\wyxn(\D_m|\X)$ is bounded from above as  
\begin{align}
\E_{\ppmx}\left(\wyxn(\D_m|\X)\right) \le \exp\left\{-(\gamma- R')\sqrt{n} + \frac{l(\xi-1)}{w} \right\} \triangleq \beta'_n, \label{eq:sj3}
\end{align}
which vanishes since $R' < \gamma$. Let $\alpha'_n \triangleq \max\{2\beta'_n, e^{-\mu\sqrt{n}/2} \}$, and note that $\alpha'_n - \beta'_n \ge e^{-\mu\sqrt{n}/2}/2$. Consider the $N$ i.i.d. random variables $\{\wyxn(\D_m|\X_{m',i})\}_{i\in[1:N]}$ which are present in the RHS of~\eqref{eq:average2}. Note that each random variable belongs to $[0,1]$, and the expectation is at most $\beta'_n$. By applying Hoeffding's inequality, we have
\begin{align}
\PP_{\{\X_{m',i}\} }\left(\frac{1}{N}\sum_{i=1}^N \wyxn(\D_m|\X_{m',i}) \ge \alpha'_n \right)
\le \exp\left\{-2N (\alpha'_n - \beta'_n)^2\right\} \le \exp\left\{-\frac{1}{2}e^{(R'-\mu)\sqrt{n}} \right\}. \label{eq:fix}
\end{align}
Note that~\eqref{eq:fix} is true for any fixed $\D_m$ (or equivalently, any fixed $\{\x_{m,j} \}_{j=1}^N$) that corresponds to  message $m$, because each of the constituent sequence in $\{\x_{m,j} \}_{j=1}^N$ is a PPM-sequence satisfying Lemma~\ref{claim:perr2}.  Next, we also take the randomness of $\{\X_{m,j} \}_{j=1}^N$ into consideration. Let $\mathbf{D}_m$ be the chance variable corresponding to $\D_m$, and we have
\begin{align}
&\PP_{\{\X_{m,i}\},\{ \X_{m',i}\} }\left(\frac{1}{N}\sum_{i=1}^N \wyxn(\mathbf{D}_m|\X_{m',i})\ge \alpha'_n \right) \label{eq:21}\\
&=\sum_{\D_m} \PP_{\{\X_{m,i}\}}(\mathbf{D}_m = \D_m)  \PP_{\{\X_{m',i}\} }\left(\frac{1}{N}\sum_{i=1}^N \wyxn(\D_m|\X_{m',i}) \ge \alpha'_n \right) \label{eq:js}\\
&\le \exp\left\{-\frac{1}{2}e^{(R'-\mu)\sqrt{n}} \right\}.\label{eq:js2}
\end{align}
Finally, a union bound over all the message pairs $(m,m') \in \M^2$ yields 
\begin{align*}
\PP\left(\max_{(m,m') \in \M^2: m\ne m'} P_{\mathrm{err}}^{(2)} \ge \alpha'_n \right) 
&\le \sum_{(m,m') \in \M^2: m\ne m'} \PP\left(\frac{1}{N}\sum_{i=1}^N  \wyxn(\mathbf{D}_{m}|\X_{m',i}) \ge \alpha'_n \right) \\
&= \exp\left\{-\frac{1}{2}e^{(R'-\mu)\sqrt{n}} + 2e^{R\sqrt{n}} \right\},
\end{align*}
which vanishes since the choice of $\mu$ ensures $R'-\mu > R$.

\begin{remark}\label{remark:compare}{\em
We now discuss why the PPM input distribution is preferable to other input distributions (such as the standard i.i.d. input distribution).  Recall from Eqn.~\eqref{eq:21} that when analyzing $P_{\mathrm{err}}^{(2)}(m,m')$ for a specific message pair $(m,m')$ in the random coding argument, one needs to take the randomness of both $\{\mathbf{X}_{m',i}\}_{i=1}^N$ (the constituent sequences of $m'$) and $\{\mathbf{X}_{m,j}\}_{j =1}^N$ (the constituent sequences of $m$) into consideration. To simplify the analysis, in Eqn.~\eqref{eq:js} we first consider the  error term $\frac{1}{N}\sum_{i=1}^N W_{Y|X}^{\otimes n}(\D_m|\X_{m',i})$ with respect to a fixed realization of the demapping region $\D_m$ (or equivalently, a fixed realization of the constituent sequences $\{\x_{m,j} \}_{j=1}^N$), and then take the randomness of $\{\mathbf{X}_{m,j}\}_{j=1}^N$ into account. 
\begin{itemize}	
\item When the PPM input distribution is used, we know that for \emph{every} realization $\D_m$, the corresponding constituent sequences $\{\x_{m,j}\}_{j=1}^N$ are all PPM-sequences that have same Hamming weight and satisfy Lemma~\ref{claim:perr2}. Thus, the error term $\frac{1}{N}\sum_{i=1}^N W_{Y|X}^{\otimes n}(\D_m|\X_{m',i})$ can be upper bounded by the \emph{same} quantity $\exp\{-\frac{1}{2}e^{(R'-\mu)\sqrt{n}} \}$ for \emph{all} realizations $\D_m$, as shown in Eqn.~\eqref{eq:fix}. Therefore, it becomes straightforward to obtain the upper bound on $P_{\mathrm{err}}^{(2)}(m,m')$ shown in Eqn.~\eqref{eq:js2}, by considering all the realizations $\D_m$ (or equivalently, all the realizations $\{\x_{m,j} \}_{j=1}^N$).
	
\item If the standard i.i.d. input distribution were used, for different realizations $\D_m$, the corresponding constituent sequences $\{\x_{m,j}\}_{j=1}^N$ would have different Hamming weights in general, thus the upper bounds on the error term $\frac{1}{N}\sum_{i=1}^N W_{Y|X}^{\otimes n}(\D_m|\X_{m',i})$ would then be different. Therefore, more effort is needed to bound $P_{\mathrm{err}}^{(2)}(m,m')$ as different realizations $\D_m$ lead to different error probabilities. Indeed, this is the reason why PPM codes are favored in the achievability parts.  
\end{itemize}
In fact, the proof technique for reliability is also applicable to \emph{any} constant composition code, i.e.,  not restricted to PPM random codes. The reason why we adopt PPM random codes is that it makes the proof of covertness easier, since, as shown in Lemma 4 to follow, the PPM-induced output distribution $P^{n,l}_{\mathbf{Z}}$ possess favorable covertness properties. }
\end{remark}

\subsubsection{Analysis of $P_{\mathrm{err}}^{(3)}$} For a fixed message $m \in \M$, the error probability of the third kind is bounded from above as 
\begin{align}
P_{\mathrm{err}}^{(3)}(m) = P_0^{\otimes n} (\D_m) \le \sum_{i=1}^N P_0^{\otimes n} (\F_{\x_{m,i}}), \notag
\end{align}
and the expected value of this error probability (averaged over the generation of $\{\X_{m,i}\}_{i=1}^N$) is bounded from above as 
\begin{align*}
\E_{\{\X_{m,i}\}}\left(P_{\mathrm{err}}^{(3)}(m)\right)
&\le \sum_{i=1}^N \E_{\{\X_{m,i}\}}\left(P_0^{\otimes n} (\F_{\X_{m,i}})\right) \\
&=e^{R'\sqrt{n}}  \sum_{\x} \ppmx(\x) \sum_{\y} P_0^{\otimes n}(\y)  \mathbbm{1}\left\{\log\frac{\wyxn(\y|\x)}{P_0^{\otimes n}(\y)} > \gamma\sqrt{n} \right\} \\
&\le e^{R'\sqrt{n}}\cdot  e^{-\gamma \sqrt{n}} \sum_{\x} \sum_{\y} \ppmx(\x) \wyxn(\y|\x) \\
&\le e^{-(\gamma-R') \sqrt{n}}.
\end{align*}
Thus, the expected value of the average error probability of the third kind satisfies 
\begin{align}
\E\left(\frac{1}{|\M|} \sum_{m \in \M} P_{\mathrm{err}}^{(3)}(m)\right) \le e^{-(\gamma-R') \sqrt{n}}.\notag
\end{align}
By applying Markov's inequality, we have 
\begin{align}
\PP\left(\frac{1}{|\M|} \sum_{m \in \M} P_{\mathrm{err}}^{(3)}(m) \ge e^{-(\gamma-R'-\mu)\sqrt{n}} \right) \le e^{-\mu\sqrt{n}}.\notag
\end{align}
That is, with probability at least $1 - e^{-\mu\sqrt{n}}$ over the  random code selection, the average error probability of the third kind $|\M|^{-1} \sum_{m \in \M} P_{\mathrm{err}}^{(3)}(m) \le e^{-(\gamma-R'-\mu)\sqrt{n}}$, which tends to zero as $n$ tends to infinity since the choice of $\mu$ ensures that $\gamma-R' > \mu$. 

\subsubsection{Analysis of covertness} First note that the KL-divergence 
\begin{align}
\DD\left(\widehat{Q}^n_{\C} \Vert Q_0^{\otimes n}\right) = \DD\left(\ppmz \Vert Q_0^{\otimes n} \right) +  \DD\left(\widehat{Q}^n_{\C} \Vert \ppmz \right) + \sum_{\z}\left(\widehat{Q}^n_{\C}(\z) - \ppmz(\z)\right)\log\frac{\ppmz(\z)}{Q_0^{\otimes n}(\z)}.\label{eq:kl}
\end{align}
In the following, we upper bound the three terms on the RHS of~\eqref{eq:kl} in Lemmas~\ref{lemma:covert1} and~\ref{lemma:covert2}.   

\begin{lemma} \label{lemma:covert1}
	For sufficiently large $n$, the KL-divergence $$\DD(\ppmz \Vert Q_0^{\otimes n}) \le \delta - \frac{1}{3}n^{-1/3}.$$
\end{lemma}
\begin{proof}[Proof of Lemma~\ref{lemma:covert1}]
	The proof is essentially due to~\cite[Lemma 1]{bloch2017optimal} and~\cite[Lemma 8]{tahmasbi2018first}, which analyze the output statistics of the PPM distribution and state that
	\begin{align}
	\DD\left(\ppmz \Vert Q_0^{\otimes n} \right) \le \frac{l^2}{2n}\chi_2(Q_1 \Vert Q_0) + \mathcal{O}\left(\frac{1}{\sqrt{n}}\right). \notag
	\end{align}
	Substituting $l = \lfloor \sqrt{(2\delta-n^{-1/3})n/\chi_2(Q_1\Vert Q_0)}  \rfloor$, we complete the proof.
\end{proof}

\begin{lemma} \label{lemma:covert2}
There exist constant $c_2, c_3 > 0$ such that with probability at least $1-\exp(-c_2\sqrt{n})$ over the random code design, the output distribution $\widehat{Q}^n_{\C}$ induced by $\C$ ensures 
	\begin{align}
	\DD\left(\widehat{Q}^n_{\C} \Vert \ppmz \right) \le \exp\{-c_3\sqrt{n} \}, \quad \mathrm{and}  \quad \sum_{\z}\left(\widehat{Q}^n(\z) - \ppmz(\z)\right)\log\frac{\ppmz(\z)}{Q_0^{\otimes n}(\z)} \le 2n\left(\log\frac{1}{\mu_0}\right)\exp\{-c_3\sqrt{n}/2\}. \notag 
	\end{align}
\end{lemma}

\begin{proof}[Proof of Lemma~\ref{lemma:covert2}]
	Recall that $\widehat{Q}^n_{\C}$ is the output distribution induced by the multiset $\cup_{m \in \M}\{\X_{m,i} \}_{i=1}^N$ with each sequence being generated i.i.d. according to $\ppmx$. We first state a result showing that the expectation of $\DD(\widehat{Q}^n_{\C} \Vert \ppmz )$ is small. 

\begin{lemma} \label{lemma:6}
Recall that $\mu_0 = \min_{z: Q_0(z)>0}Q_0(z)$, $\mu_1 = \min_{z: Q_1(z)>0}Q_1(z)$, and $\widetilde{\mu} = \min\{\mu_0, \mu_1\}$. We have 
\begin{align}
\E\left(\DD\left(\widehat{Q}^n_{\C} \Vert \ppmz \right)\right) \le \frac{e^{\tau\sqrt{n}}}{|\M|N} + 2n \log\left(1+ \widetilde{\mu} \right)e^{-c_4\sqrt{n}},
\end{align}
for some constant $c_4 > 0$.
\end{lemma} 
The proof of Lemma~\ref{lemma:6} can be found in Appendix~\ref{appendix:lemma6}.
By noting that $|\M| = \exp\{e^{R\sqrt{n}} \}$ and applying the Markov's inequality, we obtain that there exist constants $c_2, c_3 >0$ such that with probability at least $1-\exp(-c_2\sqrt{n})$ over the code design,
	\begin{align}
	&\DD\left(\widehat{Q}^n_{\C} \Vert \ppmz \right) \le \exp\{-c_3\sqrt{n} \}. \label{eq:markov}
	\end{align}
	Finally, by Pinsker's inequality, we know that 
	$$\V(\widehat{Q}^n, \ppmz) \le \sqrt{\DD(\widehat{Q}^n \Vert \ppmz)} \le \exp\{-c_3\sqrt{n}/2\},$$
	thus
	\begin{align*}
	\sum_{\z}\left(\widehat{Q}^n(\z) - \ppmz(\z)\right)\log\frac{\ppmz(\z)}{Q_0^{\otimes n}(\z)} \le 2n\left(\log\frac{1}{\mu_0}\right) \cdot \V(\widehat{Q}^n, \ppmz) \le 2n\left(\log\frac{1}{\mu_0}\right)\exp\{-c_3\sqrt{n}/2\}.
	\end{align*} 
	This completes the proof of Lemma~\ref{lemma:covert2}.

\end{proof}

\begin{remark}{\em
	As shown in~\eqref{eq:markov}, we use the Markov's inequality to show that $\DD(\widehat{Q}^n_{\C} \Vert P_{\mathbf{Z}}^{n,l}) \le \exp(-\Theta(\sqrt{n}))$ with probability at least $1 - \exp(-\Theta(\sqrt{n}))$. It is also worth pointing out that by following the finer techniques in~\cite[Lemma 2]{tahmasbi2018first}, one can obtain a stronger result in the sense that $\DD(\widehat{Q}^n_{\C} \Vert P_{\mathbf{Z}}^{n,l}) \le \exp(-\Theta(\sqrt{n}))$ with probability at least $1-\exp(-\exp(\exp(\Theta(\sqrt{n}))))$.}
\end{remark}

Combining~\eqref{eq:kl} and Lemmas~\ref{lemma:covert1} and~\ref{lemma:covert2}, we conclude that with probability at least $1-\exp(-c_2\sqrt{n})$ over the  random code  $\C$, we have 
\begin{align*}
	\DD\left(\widehat{Q}^n_{\C} \Vert Q_0^{\otimes n}\right) \le \delta
\end{align*}
for sufficiently large $n$.

\subsection{Code refinements} \label{sec:refinement}
In the following, we refine a given ``weak'' covert identification code such that the refined code satisfies the error criteria and covertness property in Definition~\ref{def:rate} and simultaneously retains the rate of the original code.

\begin{lemma} \label{lemma:rearrangement}
	Let $\delta > 0$ and $\varepsilon_n^{(1)}, \varepsilon_n^{(2)}, \varepsilon_n^{(3)}>0$ be vanishing sequences. Suppose there exists a sequence of codes $\C$ (of size $|\M|$) satisfying
	\begin{align}
	&\max_{m \in \M} P_{\mathrm{err}}^{(1)}(m)  \le \varepsilon_n^{(1)}, \  \max_{(m,m') \in \M^2: m \ne m'} P_{\mathrm{err}}^{(2)}(m,m')  \le \varepsilon_n^{(2)}, \notag\\
	&\frac{1}{|\M|}\sum_{m \in \M} P_{\mathrm{err}}^{(3)}(m) \le \varepsilon_n^{(3)}, \quad \ \DD\left(\widehat{Q}^n_{\C}\Vert Q_0^{\otimes n}\right) \le \delta. \notag
	\end{align}
	Then, there exist vanishing sequences $\widetilde{\varepsilon}_n^{(1)}, \widetilde{\varepsilon}_n^{(2)}, \widetilde{\varepsilon}_n^{(3)}>0$ (depending on $\varepsilon_n^{(1)}, \varepsilon_n^{(2)}, \varepsilon_n^{(3)}$) and another sequence of codes $\widetilde{\C}$ of size $|\widetilde{\M}| \ge (1-\widetilde{\varepsilon}_n^{(3)})|\M|$ such that 
	\begin{align}
	&\max_{m \in \widetilde{\M}} P_{\mathrm{err}}^{(1)}(m)  \le \widetilde{\varepsilon}_n^{(1)},  \  \max_{(m,m') \in \widetilde{\M}^2: m \ne m'} P_{\mathrm{err}}^{(2)}(m,m')  \le \widetilde{\varepsilon}_n^{(2)}, \notag\\
	&\max_{m \in \widetilde{\M}} P_{\mathrm{err}}^{(3)}(m) \le \widetilde{\varepsilon}_n^{(3)}, \quad \DD\left(\widehat{Q}^n_{\widetilde{\C}} \Vert Q_0^{\otimes n}\right) \le \delta. \notag
	\end{align}
\end{lemma}

\begin{proof}[Proof of Lemma~\ref{lemma:rearrangement}]
We first partition the messages in $\C$ into two disjoint sets.
\begin{definition} \label{def:good}
	Consider a code $\C$ that satisifes $\frac{1}{|\M|}\sum_{m \in \M} P_{\mathrm{err}}^{(3)}(m) \le \varepsilon_n^{(3)}$. We say a message $m \in \M$ is a \emph{good message} if $P_{\mathrm{err}}^{(3)}(m) \le (\varepsilon_n^{(3)})^{1/2}$, and a \emph{bad message} otherwise.
\end{definition}
Let $\widetilde{\M} \subset \M$ be the set that contains all the good messages, and $\widetilde{\M}^c$ be the set that contains all the bad messages. Without loss of generality, we assume $\widetilde{\M} = [1: |\widetilde{\M}| ]$ and $\widetilde{\M}^c = [|\widetilde{\M}|+1: |\M| ]$.  Since the code $\C$ satisfies $\sum_{m \in \M} P_{\mathrm{err}}^{(3)}(m) \le \varepsilon_n^{(3)} |\M|$, the number of bad messages is at most $(\varepsilon_n^{(3)})^{1/2} |\M|$, i.e.,
\begin{align}
|\widetilde{\M}^c| \le (\varepsilon_n^{(3)})^{1/2} |\M| \quad \text{and} \quad |\widetilde{\M}| \ge (1-(\varepsilon_n^{(3)})^{1/2})|\M|. \notag
\end{align} 
Recall that for each message $m \in \M$, the corresponding constituent sequences is $\{\x_{m,i} \}_{i \in [1:N]}$. We then denote the multiset of constituent sequences that correspond to all the bad messages by 
\begin{align}
\B \triangleq \cup_{m \in \widetilde{\M}^c} \{\x_{m,i} \}_{i \in [1:N]}, \notag
\end{align}
and note that $|\B| \le N \cdot (\varepsilon_n^{(3)})^{1/2} |\M|$. In the following, we construct a new code $\widetilde{\C}$ that contains $|\widetilde{\M}|$ messages. 
\begin{enumerate}
	\item We partition the multiset $\B$ into $|\widetilde{\M}|$ equal-sized disjoint subsets $\B^{(1)}, \B^{(2)}, \ldots, \B^{(|\widetilde{\M}|)}$ such that the cardinality of each subset (for $m \in \widetilde{\M}$) satisfies
	\begin{align}
	|\B^{(m)}| = \frac{|\B|}{|\widetilde{\M}|} \le \frac{N \cdot (\varepsilon_n^{(3)})^{1/2} |\M|}{(1-(\varepsilon_n^{(3)})^{1/2})|\M|} \triangleq \nu_n N, \label{eq:err3}
	\end{align}  
	where $\nu_n$ also tends to 0 as $n$ tends to infinity.
	
	\item For each $m \in \widetilde{\M}$, the corresponding multiset of constituent sequences in the original code $\C$ is $\{\x_{m,i} \}_{i \in [1:N]}$. In the new code $\widetilde{\C}$, we enlarge this multiset by appending $\B^{(m)}$ to $\{\x_{m,i} \}_{i \in [1:N]}$.  Thus, the codeword $\Um$ is the uniform distribution over a larger multiset of sequences $\{\x_{m,i} \}_{i \in [1:N]} \cup \B^{(m)}$.
	
	\item For each $m \in \widetilde{\M}$, the demapping region of the new code $\widetilde{\C}$ remains as $\D_m = \cup_{i \in [1:N]}\F_{\x_{m,i}}$. That is, the demapping regions of the new code $\widetilde{\C}$ and the original code $\C$ are exactly the same.
	
\end{enumerate}

We now analyze the error probabilities of the  new code $\widetilde{\C}$. For each $m \in \widetilde{\M}$, the error probability of the first kind is  bounded from above as 
\begin{align}
P_{\mathrm{err}}^{(1)}(m) 
 &= \frac{\sum_{i=1}^N \wyxn(\D_m^c|\x_{m,i}) + \sum_{\x \in \B^{(m)}} \wyxn(\D_m^c|\x)}{N+|\B^{(m)}|}  \notag\\
&\le \frac{N}{N+|\B^{(m)}|}\left(\frac{1}{N} \sum_{i=1}^N \wyxn(\D_m^c|\x_{m,i})\right) + \frac{|\B^{(m)}|}{N+|\B^{(m)}|} \label{eq:err1} \\
&\le \varepsilon_n^{(1)} + \nu_n, \label{eq:err2}
\end{align}
where~\eqref{eq:err1} holds since $\wyxn(\D_m^c|\x) \le 1$, and~\eqref{eq:err2} is due to~\eqref{eq:err3} and the assumption that the original code satisfies $\frac{1}{N} \sum_{i=1}^N \wyxn(\D_m^c|\x_{m,i}) \le \varepsilon_n^{(1)}$. Similarly, for each message pair $(m,m') \in \widetilde{\M}^2$, the error probability of the second kind $P_{\mathrm{err}}^{(2)}(m,m')$ is bounded from above as 
\begin{align}
P_{\mathrm{err}}^{(2)}(m,m')
&=  \frac{\sum_{i=1}^N \wyxn(\D_m|\x_{m',i}) + \sum_{\x \in \B^{(m')}} \wyxn(\D_m^c|\x)}{N+|\B^{(m')}|} \notag  \\
&\le  \frac{N}{N+|\B^{(m')}|}\left(\frac{1}{N} \sum_{i=1}^N \wyxn(\D_m^c|\x_{m',i})\right) + \frac{|\B^{(m')}|}{N+|\B^{(m')}|} \notag\\
&\le \varepsilon_n^{(2)} + \nu_n. \notag
\end{align}
Since all the messages in $\widetilde{\M}$ are good messages, by Definition~\ref{def:good} we  have that for each message $m \in \widetilde{\M}$,
\begin{align}
P_{\mathrm{err}}^{(3)}(m) \le (\varepsilon_n^{(3)})^{1/2}.  \notag
\end{align}
Finally, note that when constructing $\widetilde{\C}$, we merely rearrange the sequences of $\C$ (rather than expurgate or add any sequences); thus, the output distribution induced by $\widetilde{\C}$ is exactly the same as that induced by $\C$, i.e., 
\begin{align}
\DD\left(\widehat{Q}^n_{\widetilde{\C}} \Vert Q_0^{\otimes n}\right) = \DD\left(\widehat{Q}^n_{\C} \Vert Q_0^{\otimes n}\right) \le \delta.\notag
\end{align}
Thus, the covertness constraint is satisfied. Finally, we note that 
\begin{align*}
	\liminf_{n \to \infty} \frac{\log\log|\widetilde{\M}|}{\sqrt{n}}= 	\liminf_{n \to \infty} \frac{\log\log|\M|}{\sqrt{n}} = (1-\eta)C_{\delta},
\end{align*}
and the proof is completed by taking $\eta\to 0^+$. 
\end{proof}

\subsection{Discussions of other identification schemes} \label{sec:other}
Various achievability schemes have been developed for the standard (non-covert) identification problem. 
In particular, we note that the approach described in~\cite{ahlswede1989identification-feedback,ahlswede2008general} draws a clear connection between the identification and transmission problems from the achievability's perspective, via the agreement of a shared key (see~\cite{sudan2019communication} for a survey). The key idea there is that (i) in the first stage, the sender and receivers use a capacity-achieving transmission code to communicate a message of size $\exp(n\mathsf{C}_W)$ as the shared key (over approximately $n$ channel uses), and (ii) in the second stage, they identify an ID message using another transmission code with a negligible blocklength (i.e., over $o(n)$ channel uses). It turns out that with a shared key of size $\exp(n\mathsf{C}_W)$ obtained in Stage 1 and another $o(n)$ channel uses, the size of the ID message that can be correctly identified is approximately exponential in the size of the shared key. From this approach, it is then clear that the identification capacity is at least as large as the channel capacity for the standard identification problem.

However, this approach is strictly sub-optimal in the presence of covertness constraints. This is because the sender and receivers are not assumed to have a shared key prior to communication, and in this scenario there does not always exist a covert-capacity-achieving transmission code. Specifically, when the channels $\wyx$ and $\wzx$ satisfy $\DD(P_1 \Vert P_0) \le \DD(Q_1 \Vert Q_0)$ (i.e., $\wzx$ is ``better'' than $\wyx$), a shared key is necessary for constructing a covert-capacity-achieving transmission code, thus the approach in~\cite{ahlswede1989identification-feedback,ahlswede2008general} is not applicable.  On the other hand, when the channels $\wyx$ and $\wzx$ satisfy $\DD(P_1 \Vert P_0) > \DD(Q_1 \Vert Q_0)$, it is then possible to adapt this approach to the covert identification problem, by using a covert-capacity-achieving transmission code in the first stage and another covert transmission code in the second stage. With some technical arguments in~\cite{ahlswede1989identification-feedback,ahlswede2008general}, it is not difficult to analyze the error probabilities for identification; however, proving covertness may become non-trivial since one needs to analyze the concatenation of codes defined for each of the two stages, and also take into account the dependence of the channel outputs corresponding to the two stages.

\section{Converse} \label{sec:converse}

In this section, we show that any sequence of identification codes with size $|\M|$ that simultaneously guarantees that $\DD(\widehat{Q}^n_{\C} \Vert Q_0^{\otimes n}) \le \delta$ and $P_{\mathrm{err}}^{(1)} = \lambda_n^{(1)}, P_{\mathrm{err}}^{(2)}= \lambda_n^{(2)}, P_{\mathrm{err}}^{(3)} = \lambda_n^{(3)}$ (where $\lim_{n \to \infty} \lambda_n^{(1)} = \lim_{n \to \infty} \lambda_n^{(2)} = \lim_{n \to \infty} \lambda_n^{(3)} = 0$) must satisfy 
\begin{align}
\limsup_{n \to \infty} \frac{\log\log |\M|}{\sqrt{n}} \le C_{\delta}.\notag
\end{align}

\begin{lemma} \label{lemma:constraint}
	Consider any identification code $\C$ with message set $\M$, codewords $\{U_m \}_{m \in \M}$, and demapping regions $\{\D_m \}_{m \in \M}$ such that $\DD(\widehat{Q}^n_{\C} \Vert Q_0^{\otimes n}) \le \delta$. Let $f_{\mathrm{H}}(m)\triangleq \frac{1}{n}\sum_{\x}\Uxm\mathrm{wt}_{\mathrm{H}}(\x)$ be the fractional Hamming weight for each message $m \in \M$. Then, there exists a constant $c_5 > 0$ such that the average fractional Hamming weight of $\C$ satisfies
	\begin{align}
	\frac{1}{|\M|} \sum_{m \in \M} f_{\mathrm{H}}(m) \le  \sqrt{\frac{2\delta}{\chi_2(Q_1 \Vert Q_0)}}\left(\frac{1}{\sqrt{n}} + \frac{c_5}{n} \right). \label{eq:ave}
	\end{align}
\end{lemma}

\begin{proof}[Proof of Lemma~\ref{lemma:constraint}]
	We denote the $i$-th marginal distribution of each codeword $U_m$ as $(\Um)_i$ for $i \in [1:n]$, and the $i$-marginal distribution of $\widehat{Q}^n_{\C}$ as $(\widehat{Q}^n_{\C})_i$, which takes the form
	\begin{align}
	(\widehat{Q}^n_{\C})_i(z) = \frac{1}{|\M|}\sum_{m \in \M} \sum_{x \in \mathcal{X}} (\Um)_i(x) \wzx(z|x), \ \forall z \in \mathcal{Z}. \notag
	\end{align}
	Let $\bar{Q}_{\C}(z) \triangleq \frac{1}{n}\sum_{i=1}^n (\widehat{Q}^n_{\C})_i(z)$. By taking the covertness constraint into account and following the analysis in~\cite[Eqn.~(13)]{wang2016fundamental}, we have
	\begin{align}
	\delta \ge \DD\left(\widehat{Q}^n_{\C} \Vert Q_0^{\otimes n}\right) \ge n \DD\left(\bar{Q}_{\C} \Vert Q_0\right), \label{eq:fy1}
	\end{align}
	and thus $\lim_{n \to \infty} \DD\left(\bar{Q}_{\C} \Vert Q_0\right) = 0$. By applying Pinsker's inequality $\V\left(\bar{Q}_{\C}, Q_0\right) \le \sqrt{\DD\left(\bar{Q}_{\C} \Vert Q_0\right)/2}$, we also have $$\lim_{n \to \infty}\V\left(\bar{Q}_{\C}, Q_0\right) = 0.$$
	Let $\psi =\psi_n\triangleq \frac{1}{n}\sum_{i=1}^n \frac{1}{|\M|}\sum_{m \in \M}  (\Um)_i(1)$ be the fraction of $1$'s in the codebook, and one can express $\bar{Q}_{\C}(z)$ as
	\begin{align*}
	\bar{Q}_{\C}(z) &= \frac{1}{n}\sum_{i=1}^n \frac{1}{|\M|}\sum_{m \in \M} \sum_{x \in \mathcal{X}} (\Um)_i(x) \wzx(z|x) \\
	&= \left( \frac{1}{n}\sum_{i=1}^n \frac{1}{|\M|}\sum_{m \in \M}  (\Um)_i(1)\right) Q_1(z) + \left( \frac{1}{n}\sum_{i=1}^n \frac{1}{|\M|}\sum_{m \in \M} (\Um)_i(0)\right) Q_0(z) \\
	&= \psi Q_1(z) + (1-\psi) Q_0(z).
	\end{align*}
	Note that the requirement on variational distance $\lim_{n \to \infty}\V\left(\bar{Q}_{\C}, Q_0\right) = 0$ implies that $\lim_{n \to \infty} \psi = 0$. Furthermore, we know from~\cite[Eqn. (11)]{bloch2016covert} that 
	\begin{align}
	\DD\left(\bar{Q}_{\C} \Vert Q_0\right) \ge \frac{\psi^2}{2}\chi_2(Q_1 \Vert Q_0) - \mathcal{O}(\psi^3). \label{eq:fy2}
	\end{align} 
	Combining~\eqref{eq:fy1} and~\eqref{eq:fy2}, one can bound $\psi$ from above as
	\begin{align}
	\psi \le \sqrt{\frac{2\delta}{\chi_2(Q_1 \Vert Q_0)}}\left(\frac{1}{\sqrt{n}} + \frac{c_5}{n} \right), \label{eq:tian1}
	\end{align} 
	for some constant $c_5 > 0$. At the same time, one also can interpret $\psi$ as the average fractional Hamming weight of the code, since 
	\begin{align}
	\psi &= \frac{1}{n}\sum_{i=1}^n \frac{1}{|\M|}\sum_{m \in \M}  (\Um)_i(1) \notag\\
	&= \frac{1}{|\M|}\sum_{m \in \M} \frac{1}{n}\sum_{i=1}^n \sum_{x_i}  (\Um)_i(x_i) \mathbbm{1}\left\{x_i = 1 \right\} \notag\\
	&= \frac{1}{|\M|}\sum_{m \in \M} \frac{1}{n}\sum_{i=1}^n \sum_{x_i}\left( \sum_{x^{(-i)}} U_m(x_i, x^{(-i)})\right) \mathbbm{1}\left\{x_i = 1 \right\}\notag \\
	&= \frac{1}{|\M|}\sum_{m \in \M} \frac{1}{n}\sum_{i=1}^n \sum_{\x} \Uxm \mathbbm{1}\left\{x_i = 1 \right\} \notag\\
	&=\frac{1}{|\M|}\sum_{m \in \M} \frac{1}{n} \sum_{\x} U_m(\x) \mathrm{wt}_{\mathrm{H}}(\x) \notag\\
	&= \frac{1}{|\M|} \sum_{m \in \M} f_{\mathrm{H}}(m), \label{eq:tian2}
	\end{align}
	where $x^{(-i)} = (x_1,\ldots, x_{i-1}, x_{i+1},\ldots, x_n)\in\mathcal{X}^{n-1}$. This completes the proof of Lemma~\ref{lemma:constraint}.  
\end{proof}

For notational convenience, let $$k \triangleq \sqrt{\frac{2\delta}{\chi_2(Q_1 \Vert Q_0)}}\left(\frac{1}{\sqrt{n}} + \frac{c_5}{n} \right).$$

\begin{lemma}[Expurgation Lemma] \label{lemma:expurgation}
	Suppose there exists a sequence of identification codes $\C$ with message set $\M$, codewords $\{U_m \}_{m \in \M}$, and demapping regions $\{\D_m \}_{m \in \M}$ such that $\DD(\widehat{Q}^n_{\C} \Vert Q_0^{\otimes n}) \le \delta$, $P_{\mathrm{err}}^{(1)} = \lambda_n^{(1)}$, $P_{\mathrm{err}}^{(2)} = \lambda_n^{(2)}$, and $P_{\mathrm{err}}^{(3)} = \lambda_n^{(3)}$, where $\lim_{n \to \infty} \lambda_n^{(1)} = \lim_{n \to \infty} \lambda_n^{(2)} = \lim_{n \to \infty} \lambda_n^{(3)} = 0$. 
	
	Then, there exist a sequence $\kappa_n > 0$ (which  depends on $\lambda_n^{(1)},\lambda_n^{(2)}$) which satisfies $\lim_{n \to \infty} \kappa_n = 0$ and a sequence of identification codes $\C'$ with message set $\M'$, codewords $\{U'_m \}_{m \in \M'}$, and demapping regions $\{\D'_m \}_{m \in \M'}$ such that 
	\begin{enumerate}
		\item $|\M'| \ge |\M|/(n+1)$;
		\item For every $m \in \M'$, $U'_{m}(\x) = 0$ for all $\x$ such that $\mathrm{wt}_{\mathrm{H}}(\x) >(1+\kappa_n)kn$;
		\item $P_{\mathrm{err}}^{(1)} \le (\lambda_n^{(1)})^{1/2}$, $P_{\mathrm{err}}^{(2)} \le (\lambda_n^{(2)})^{1/2}$, and $P_{\mathrm{err}}^{(3)} \le\lambda_n^{(3)}$.
	\end{enumerate}
\end{lemma}

\begin{proof}[Proof of Lemma~\ref{lemma:expurgation}]
	Since the identification code $\C$ satisfies $\DD(\widehat{Q}^n_{\C} \Vert Q_0^{\otimes n}) \le \delta$, Lemma~\ref{lemma:constraint} above ensures that its average fractional Hamming weight $\frac{1}{|\M|} \sum_{m \in \M} f_{\mathrm{H}}(m) \le k$. We define $\mathcal{G}$ as the subset of messages with small fractional Hamming weight, i.e., 
	\begin{align}
	\mathcal{G} \triangleq \left\{m \in \M :f_{\mathrm{H}}(m) \le  \left(1+\frac{1}{n}\right)k \right\}. \label{eq:g}
	\end{align} 
	From~\eqref{eq:tian1} and~\eqref{eq:tian2}, we have
	\begin{align}
	k \ge \psi = \frac{1}{|\M|} \sum_{m \in \mathcal{G}} f_{\mathrm{H}}(m) + \frac{1}{|\M|} \sum_{m \in \M \setminus \mathcal{G}} f_{\mathrm{H}}(m) \ge \frac{|\M \setminus\mathcal{G}|}{|\M|} \left(1+\frac{1}{n}\right)k, \notag
	\end{align}
	which further implies that $|\mathcal{G}| \ge |\M|/(n+1)$, i.e., the number of messages with small fractional Hamming weight is not small. Let $\lambda_n \triangleq \max\{\lambda^{(1)}_n, \lambda^{(2)}_n \}$ and  $\epsilon_n \triangleq \frac{\sqrt{\lambda_n}}{1-\sqrt{\lambda_n}}$. We partition $\mathcal{X}^n$ into two disjoint sets---the \emph{low-weight set} $\mathcal{X}^n_{\mathrm{l}} \triangleq \{\x\in \mathcal{X}^n: \mathrm{wt}_{\mathrm{H}}(\x) \le (1+\epsilon_n)\left(1+\frac{1}{n}\right) kn  \}$ and the \emph{high-weight set} $\mathcal{X}^n_{\mathrm{h}} \triangleq \mathcal{X}^n \setminus \mathcal{X}^n_{\mathrm{l}}$.  In the following, we describe the procedure of constructing the new code $\C'$.
	\begin{enumerate}
		\item First, the message set of the new code is $\M' = \mathcal{G}$. Thus, $f_{\mathrm{H}}(m) \le  \left(1+\frac{1}{n}\right)k$ for all $m \in \M'$.
		\item For each $m \in \M'$, we define $g_m \triangleq \sum_{\x \in \mathcal{X}^n_{\mathrm{l}}}U_m(\x)$, and we set the codeword $U'_m$ of the new code $\C'$ to be 
		\begin{align}
		U'_m(\x) = \begin{cases}
		U_m(\x)/g_m, &\text{if }  \x \in \mathcal{X}^n_{\mathrm{l}}, \\
		0, &\text{otherwise}.  \notag
		\end{cases}
		\end{align}
		One can check that $\sum_{\x} U'_m(\x) = 1$.
		\item The demapping regions of the new code $\C'$ are the same as those of $\C$, i.e., $\D'_m = \D_m$ for all $m \in \M'$.
	\end{enumerate}
	
	From~\eqref{eq:g} we have that for each $m \in \M'$,
	\begin{align}
	\left(1+\frac{1}{n}\right)k \ge f_{\mathrm{H}}(m) 
	&= \frac{1}{n}\sum_{\x} U_m(\x)\text{wt}_{\mathrm{H}}(\x) \notag\\
	&\ge \frac{1}{n}\sum_{\x \in \mathcal{X}^n_{\mathrm{h}} } U_m(\x)\text{wt}_{\mathrm{H}}(\x) \notag\\
	&\ge \frac{1}{n}\sum_{\x \in \mathcal{X}^n_{\mathrm{h}} } U_m(\x)\cdot (1+\epsilon_n)\left(1+\frac{1}{n}\right)kn, \notag
	\end{align}
	which yields a lower bound on $g_m$, i.e., 
	\begin{align}
	g_m = \sum_{\x \in \mathcal{X}^n_{\mathrm{l}} } U_m(\x) =  1 -   \sum_{\x\in \mathcal{X}^n_{\mathrm{h}} } U_m(\x) \ge \frac{\epsilon_n}{1+\epsilon_n}. \label{eq:en3} 
	\end{align}
	
	We now analyze the error probabilities of the new code $\C'$ which only consists of low-weight sequences. For each $m \in \M'$, the error probability of the first kind $P_{\mathrm{err}}^{(1)}(m)$ can be bounded from above as 
	\begin{align}
	P_{\mathrm{err}}^{(1)}(m) \notag
	&= \sum_{\x\in \mathcal{X}^n_{\mathrm{l}}} U'_m(\x) \wyxn(\D_m|\x) \\
	&= \sum_{\x\in \mathcal{X}^n_{\mathrm{l}}} \frac{U_m(\x)}{g_m} \wyxn(\D_m|\x)\notag\\
	 &\le \frac{1}{g_m} \sum_{\x} U_m(\x)\wyxn(\D_m|\x) \notag\\
	&\le \left(\frac{1+\epsilon_n}{\epsilon_n}\right) \lambda_n^{(1)} \label{eq:en1} \\
	&\le \left(\lambda_n^{(1)}\right)^{1/2}, \label{eq:en2}
	\end{align}
	where~\eqref{eq:en1} follows from~\eqref{eq:en3} the fact that the original code $\C$ satisfies $\sum_{\x} U_m(\x)\wyxn(\D_m|\x) \le \lambda_n^{(1)}$, and~\eqref{eq:en2} is due to the choice of $\epsilon_n$. Furthermore, for each message pair $(m,m') \in \M' \times \M'$ such that $m \ne m'$, the error probability of the second kind $P_{\mathrm{err}}^{(2)}(m,m')$ can be similarly bounded from above as 
	\begin{align*}
	P_{\mathrm{err}}^{(2)}(m,m') &= \sum_{\x \in \mathcal{X}^n_{\mathrm{l}}} U'_{m'}(\x) \wyxn(\D_m^c|\x) \\
	&= \sum_{\x\in \mathcal{X}^n_{\mathrm{l}}} \frac{U_{m'}(\x)}{g_{m'}} \wyxn(\D_m^c|\x) \\
	&= \frac{1}{g_{m'}} \sum_{\x}U_{m'}(\x) \wyxn(\D_m^c|\x) \\
	&\le  \left(\frac{1+\epsilon_n}{\epsilon_n}\right) \lambda_n^{(2)} \\
	&\le \left(\lambda_n^{(2)}\right)^{1/2}.
	\end{align*}
	Finally, we note that the error probability of the third kind $P_{\mathrm{err}}^{(3)}(m) = P_0^{\otimes n}(\D'_m)$ is still bounded from above by $\lambda_n^{(3)}$, since the demapping regions are unchanged, i.e.,  $\D'_m = \D_m$ for $m \in \M'$. We complete the proof of  Lemma~\ref{lemma:expurgation} by setting $\kappa_n = (1+\epsilon_n)\left(1+\frac{1}{n}\right) - 1$, which vanishes as $n$ tends to infinity.
\end{proof}

Proving the converse of identification problems usually relies on the achievability results for the channel resolvability problem. In the following, we first introduce the definition of the \emph{$K$-type distributions}, and then state a modified version of the channel resolvability result in Lemma~\ref{lemma:CR}. Lemma~\ref{lemma:CR} is modified from the so-called \emph{soft-covering lemma} presented by Cuff~\cite[Corollary VII.2]{cuff2013distributed}.

\begin{definition} \label{def:type}
	For any positive integer $K$, a probability
	distribution $P \in \mathcal{P}(\mathcal{X})$ is said to be a {\em $K$-type distribution} if 
	\begin{align}
	P(x) \in \left\{0, \frac{1}{K},\frac{2}{K}, \ldots, 1 \right\}, \quad \forall x \in \mathcal{X}.\notag
	\end{align}
\end{definition}

\begin{lemma} \label{lemma:CR}
	Let $P_{\X} \in \mathcal{P}(\mathcal{X}^n)$ and $P_{\Y}(\y) = \sum_{\x}P_{\X}(\x)\wyxn(\y|\x)$. We randomly sample $K$ i.i.d.\ sequences $\x_1, \ldots, \x_K$ according to $P_{\X}$. Let $$\widetilde{P}_{\X}(\x) = \frac{1}{K}\sum_{i=1}^K \mathbbm{1}\{\x = \x_i \}, \quad \forall \x \in \mathcal{X}^n$$ be a $K$-type distribution and $\widetilde{P}_{\Y}(\y) = \sum_{\x}\widetilde{P}_{\X}(\x)\wyxn(\y|\x)$ be the corresponding output distribution. Then, for any $\zeta > 0$ and any $P'_{\Y} \in \mathcal{P}(\mathcal{Y}^n)$, 
	\begin{align}
	\E\left(\V\left(P_{\Y}, \widetilde{P}_{\Y}  \right) \right) 
	&\le \PP_{P_{\X}\wyxn}\left(\log\frac{\wyxn(\Y|\X)}{P'_{\Y}(\Y)} > \zeta \right) + \frac{1}{2} \sqrt{\frac{e^{\zeta}}{K}},\notag
	\end{align} 
	where the expectation on the left-hand-side of the above inequality  is over the random generation  of $\x_1, \ldots, \x_K$.
\end{lemma}
\begin{proof}
	The proof is presented in Appendix~\ref{appendix:resolvability}, and is adapted from~\cite[Section VII-C]{cuff2013distributed} with appropriate modifications.
\end{proof}
Note that Lemma~\ref{lemma:CR} above holds for {\em any} $P'_{\Y} \in \mathcal{P}(\mathcal{Y}^n)$, which differs from an analogous (but more restrictive) result in~\cite[Corollary~VII.2]{cuff2013distributed} wherein $P'_{\Y}$ is set to be $P_{\Y}$. This flexibility of choosing $P'_{\Y}$ arbitrarily is important for proving the converse because we need to set it to $P_0^{\otimes n}$ later for the analysis of the covert identification problem. We now consider the identification code $\C'$ constructed in Lemma~\ref{lemma:expurgation}.

\begin{lemma} \label{lemma:approximation}
	Let $\log K \triangleq \lceil (1+n^{-1/6})^2(1+\kappa_n)kn\DD(P_1 \Vert P_0) \rceil$. For every message $m \in \M'$ with codeword $U'_m$, there exists a $K$-type distribution $\tUm$ such that 
	\begin{align}
	\V\left(U'_m\wyxn, \tUm\wyxn \right) \le \exp \big(-c_6 n^{1/6} \big) \notag
	\end{align}
	for some constant $c_6 > 0$, where $U'_m\wyxn$ and $\tUm\wyxn$ respectively denote the distributions on $\mathcal{Y}^n$ induced by $U'_m$ and $\tUm$ through the channel $\wyxn$.
\end{lemma}
\begin{proof}[Proof of Lemma~\ref{lemma:approximation}]
	Consider a specific $m \in \M'$ with codeword $U'_m$. Substituting $P_{\X}$ with $U'_m$, $P'_{\Y}$ with $P_0^{\otimes n}$, and setting $\zeta \triangleq (1+n^{-1/6})(1+\kappa_n)kn\DD(P_1 \Vert P_0)$ in Lemma~\ref{lemma:CR}, we have 
	\begin{align}
	\PP_{U'_m\wzxn}\left(\log\frac{\wyxn(\Y|\X)}{P_0^{\otimes n}(\Y)} > \zeta \right) 
	&= \sum_{\x}\sum_{\y} U'_m(\x)\wyxn(\y|\x) \mathbbm{1}\left(\log\frac{\wyxn(\y|\x)}{P_0^{\otimes n}(\y)} > \zeta \right)\notag \\
	&= \sum_{q = 0}^{(1+\kappa_n)kn} \sum_{\x: \text{wt}_{\mathrm{H}}(\x) = q}  U'_m(\x) \sum_{\y} \wyxn(\y|\x)\times \mathbbm{1}\left(\log\frac{\wyxn(\y|\x)}{P_0^{\otimes n}(\y)} > \zeta \right) \label{eq:weight} \\
	&=\sum_{q = 0}^{(1+\kappa_n)kn} \sum_{\x: \text{wt}_{\mathrm{H}}(\x) = q} U'_m(\x)\ \PP_{P_1^{\otimes q}}\left(\sum_{i=1}^q \log\frac{P_1(Y_i)}{P_0(Y_i)} > \zeta \right), \label{eq:weight2}
	\end{align}     
	where in~\eqref{eq:weight} we partition $\x$ into different type classes characterized by their Hamming weights, and~\eqref{eq:weight2} is obtained by assuming $x_i = 1$ for $i \in [1:q]$ and $x_{i} = 0$ for $i \in [q+1:n]$ without loss of generality. Also note that
	\begin{align}
		\zeta - q\DD(P_1 \Vert P_0) \ge \zeta - (1+\kappa_n)kn\DD(P_1 \Vert P_0) = n^{-1/6}(1+\kappa_n)kn \DD(P_1 \Vert P_0) \triangleq \Upsilon. \label{eq:U}
	\end{align} 
Thus, we have
	\begin{align}
		\PP_{P_1^{\otimes q}}\left(\sum_{i=1}^q \log\frac{P_1(Y_i)}{P_0(Y_i)} > \zeta \right) 
		&\le \PP_{P_1^{\otimes q}}\left(\sum_{i=1}^q \log\frac{P_1(Y_i)}{P_0(Y_i)} - q\DD(P_1 \Vert P_0) > \Upsilon \right) \label{eq:sh} \\
		&\le \exp\big(-c_7 n^{1/6} \big), \label{eq:hoe}
	\end{align}
where~\eqref{eq:sh} is obtained by subtracting $q\DD(P_1 \Vert P_0)$ from both sides and by the inequality in~\eqref{eq:U}, and~\eqref{eq:hoe} holds for some constant $c_7 > 0$ and is obtained by applying Hoeffding's inequality. Hence, the term in~\eqref{eq:weight2} is bounded from above by $\exp\left(-c_7 n^{1/6} \right)$. Furthermore, one can also show that $\sqrt{e^{\zeta}/K} \le \exp (-n^{1/6}\zeta )$. 

Therefore, by Lemma~\ref{lemma:CR}, for every message $m \in \M'$, there exists a $K$-type distribution $\tUm$ such that 
	\begin{align}
	\V\left(U'_m\wyxn, \tUm\wyxn \right) \le \exp \big(-c_7 n^{1/6} \big) +  \exp\big(-n^{1/6}\zeta \big) 
	\le \exp\big(-c_6 n^{1/6}\big),\notag
	\end{align}
	for some constant $c_6 > 0$ and all $n$ large enough. 
\end{proof}
In the following, we apply standard channel identification converse techniques to the code $\C'$. For any $m,m' \in \M'$ such that $m \ne m'$, we have 
\begin{align}
\V\left(U'_m \wyxn, U'_{m'} \wyxn \right) \ge U'_m \wyxn(\D_m) - U'_{m'} \wyxn(\D_m) \ge 1-(\lambda^{(1)}_n)^{1/2} - (\lambda^{(2)}_n)^{1/2}, \label{eq:contradiction}
\end{align} 
where the last inequality is due to Lemma~\ref{lemma:expurgation} which states that the error probabilities of $\C'$ satisfies $P_{\mathrm{err}}^{(1)} \le (\lambda_n^{(1)})^{1/2}$ and $P_{\mathrm{err}}^{(2)} \le (\lambda_n^{(2)})^{1/2}$.
Meanwhile, from Lemma~\ref{lemma:approximation} we know that there exists a set of  $K$-type distributions $\{\tUm \}_{m \in \M'}$ such that 
\begin{align}
\V(U'_m \wyxn, \tUm \wyxn) \le \exp\left(-c_6 n^{1/6} \right), \ \forall m \in \M'. \label{eq:contradiction2}
\end{align}
Combining~\eqref{eq:contradiction} and~\eqref{eq:contradiction2}, we have the following claim.
\begin{lemma} \label{claim}
	For sufficiently large $n$,	the distributions in $\{\tUm \}_{m \in \M'}$ are distinct, i.e., there does not exist $(m,m')$ with $m \ne m'$ such that $\tUm = \widetilde{U}_{m'}$.
\end{lemma}
\begin{proof}[Proof of Lemma~\ref{claim}]
	Suppose $\tUm = \widetilde{U}_{m'}$ for some $m \ne m'$. By the triangle inequality, we have 
	\begin{align}
	\V(U'_m \wyxn, U'_{m'} \wyxn)
	&\le \V(U'_m \wyxn, \tUm \wyxn) + \V(\tUm \wyxn, U'_{m'}\wyxn)\notag\\
	&= \V(U'_m \wyxn, \tUm \wyxn) + \V(\widetilde{U}_{m'}\wyxn, U'_{m'}\wyxn) \notag\\
	&\le 2\exp \big(-c_6 n^{1/6} \big),\notag
	\end{align}
	which contradicts~\eqref{eq:contradiction} for sufficiently large $n$.
\end{proof}
It is worth noting that the number of distinct $K$-type distributions on $\mathcal{X}^n$ is at most $|\mathcal{X}|^{nK}$. Thus, combining Lemma~\ref{lemma:approximation} and Lemma~\ref{claim}, we have 
\begin{align}
|\M'| \le |\mathcal{X}|^{nK},\notag
\end{align}
and by taking iterated logarithms on both sides, we have 
\begin{align}
\log \log |\M'| \le \log K + \log n + \log \log |\mathcal{X}|.\notag
\end{align}
Therefore, by recalling  that $|\M'| \ge |\M|/(n+1)$, $\log K = \lceil (1+n^{-1/6})^2(1+\kappa_n)kn\DD(P_1 \Vert P_0) \rceil$, and $k = \sqrt{\frac{2\delta}{\chi_2(Q_1 \Vert Q_0)}}\left(\frac{1}{\sqrt{n}} + \frac{c_5}{n} \right)$, we eventually obtain that 
\begin{align*}
\limsup_{n \to \infty} \frac{\log \log |\M|}{\sqrt{n}}
&\le \limsup_{n \to \infty} \frac{\log \log |\M'|}{\sqrt{n}} \\
&\le \limsup_{n \to \infty} \left(\frac{\log K}{\sqrt{n}} + \frac{\log n}{\sqrt{n}} + \frac{\log \log |\mathcal{X}|}{\sqrt{n}}  \right) \\
&= \sqrt{\frac{2\delta}{\chi_2(Q_1 \Vert Q_0)}} \DD(P_1 \Vert P_0) = C_{\delta}.
\end{align*}
This completes the proof of the converse part.

\section{Concluding Remarks} \label{sec:conclusion}

This work investigates the covert identification problem over binary-input discrete memoryless channels, showing that an ID message of size $\exp(\exp(\Theta(\sqrt{n})))$ can be reliably and covertly transmitted over $n$ channel uses. We also characterize the covert identification capacity and show that it equals the covert capacity in the standard covert communication problem. The covert identification capacity  can  be achieved  without any shared key.  

Finally, we put forth several directions that we believe are fertile avenues for future research.
\begin{itemize}
	\item Strictly speaking, the converse result established in Section~\ref{sec:converse} is commonly known as a \emph{weak converse} because all three error probabilities are allowed to vanish as $n$ grows. One would then expect that a \emph{strong converse} for the covert identification problem can be shown. This can perhaps be achieved following the lead of~\cite{han1992new} and~\cite[Chapter 6]{koga2013information} for the standard channel identification problem. The key limitation of our converse technique that prevents us from deriving the strong converse is the use of Lemma~\ref{lemma:expurgation} (Expurgation Lemma), wherein we expurgate  many high-weight sequences such that the error probabilities of the expurgated code increase significantly. Thus, a promising way to circumvent this issue might be developing a more general result for channel resolvability with stringent input constraints (i.e., extending the applicability of Lemma~\ref{lemma:approximation}), instead of applying the Expurgation Lemma.
	
	\item Having established the (first-order) fundamental limits, it is then natural to derive the \emph{error exponent} of the covert identification problem. One may follow the lead of the error exponent analysis for the standard identification problem by Ahlswede and Dueck~\cite{ahlswede1989identification}. However, due to the stringent input constraints mandated by the covertness constraints, this strategy requires special care and new analytical techniques to obtain closed-form expressions.
	
	\item In addition to the KL-divergence metric studied in this work, it is also worth considering alternative covertness metrics such as the variational distance and the probability of missed detection~\cite{tahmasbi2018first}. 
\end{itemize}

\appendices

\section{Proof of Lemma~\ref{claim:perr2}} \label{appendix:lemma3}
	For any PPM-sequence $\widetilde{\x} \in \mathcal{X}^n$, we have
	\begin{align}
	\E_{\ppmx}\left(\wyxn(\F_{\widetilde{\x}}|\X)\right)
	&= \sum_{\x}\ppmx(\x) \sum_{\y}\wyxn(\y|\x)\mathbbm{1}\left\{\log\frac{\wyxn(\y|\widetilde{\x})}{P_0^{\otimes n}(\y)} > \gamma\sqrt{n} \right\} \notag\\
	&= \sum_{\y} \frac{\ppmy(\y)}{P_0^{\otimes n}(\y)} P_0^{\otimes n}(\y) \mathbbm{1}\left\{\log\frac{\wyxn(\y|\widetilde{\x})}{P_0^{\otimes n}(\y)} > \gamma\sqrt{n} \right\} \notag\\
	&\le e^{-\gamma\sqrt{n}} \sum_{\y} \frac{\ppmy(\y)}{P_0^{\otimes n}(\y)} \wyxn(\y|\widetilde{\x}) \label{eq:ha1} \\
	&= e^{-\gamma\sqrt{n}} \sum_{\y} \frac{\left(\prod_{i=1}^l \pwy(\uy^{(i)})\right) \cdot P_0^{\otimes s}(y_{wl+1}^n)}{P_0^{\otimes n}(\y)}\wyxn(\y|\widetilde{\x}) \label{eq:ha2}\\
	&= e^{-\gamma\sqrt{n}} \left(\prod_{i=1}^l \sum_{\uy^{(i)}} \frac{\pwy(\uy^{(i)})}{P_0^{\otimes w}(\uy^{(i)})} \wyxw(\uy^{(i)}|\widetilde{\ux}^{(i)}) \right)  \left(\sum_{y_{wl+1}^n} \frac{P_0^{\otimes s}(y_{wl+1}^n)}{P_0^{\otimes s}(y_{wl+1}^n)}P_0^{\otimes s}(y_{wl+1}^n) \right), \label{eq:stand2}
	\end{align}
	where~\eqref{eq:ha1} holds since we only consider $\y$ that satisfies $\log\left(\wyxn(\y|\widetilde{\x})/ P_0^{\otimes n}(\y)\right) > \gamma\sqrt{n}$, and~\eqref{eq:ha2} follows from~\eqref{eq:mean}. Without loss of generality, we consider the first interval $[1:w]$ such that $\uy^{(1)} = [y_1, \ldots, y_w]$ and $\widetilde{\ux}^{(1)} = [\widetilde{x}_1, \ldots, \widetilde{x}_w]$, and by symmetry we further assume $\widetilde{x}_1 = 1$ and $\widetilde{x}_j = 0$ for $j \in [2:w]$. Thus, 
	\begin{align}
	\sum_{\uy^{(1)}} \frac{\pwy(\uy^{(1)})}{P_0^{\otimes w}(\uy^{(1)})} \wyxw(\uy^{(1)}|\widetilde{\ux}^{(1)}) 
	&= \sum_{\uy^{(1)}} \frac{\pwy(\uy^{(1)})}{P_0^{\otimes w}(\uy^{(1)})} \left(P_1(y_1) \prod_{j=2}^w P_0(y_j)\right)	\notag\\
	&= \sum_{y_1}\frac{(\pwy)_1(y_1)}{P_0(y_1)}P_1(y_1) \label{eq:july}\\
	&= \sum_{y_1}\frac{\frac{1}{w}P_1(y_1)+\frac{w-1}{w}P_0(y_1)}{P_0(y_1)}P_1(y_1) \notag\\
	&= 1 + \frac{1}{w}(\xi-1), \label{eq:stand}
	\end{align}
	where $(\pwy)_1$ in~\eqref{eq:july} stands for the marginal distribution of $\pwy$ which takes the form $\frac{1}{w}P_1+\frac{w-1}{w}P_0$. 
	Combining~\eqref{eq:stand2} and~\eqref{eq:stand} and applying the inequality $\log(1+x) \le x$, we have
	\begin{align}
	\E_{\ppmx}\left(\wyxn(\F_{\widetilde{\x}}|\X)\right) \le e^{-\gamma\sqrt{n}}\left(1 + \frac{1}{w}(\xi-1)\right)^l \le e^{-\gamma\sqrt{n}} e^{l(\xi-1)/w},\notag
	\end{align}
	which completes the proof.

\section{Proof of Lemma~\ref{lemma:6}} \label{appendix:lemma6}
	We first borrow a result from~\cite[Eq. (10)]{hou2013informational} and~\cite[Eq. (81)]{tahmasbi2018first} which states that  
\begin{align}
\E\left(\DD\left(\widehat{Q}^n_{\C} \Vert \ppmz \right)\right) \le \E_{\ppmx\wzxn}\left( \log\left(1+\frac{\wzxn(\Z|\X)}{|\M|N \ppmz(\Z)} \right)\right). \label{eq:combine1}
\end{align}
Let $\tau \triangleq 2t\DD(Q_1 \Vert Q_0)$ and $$\B_{\tau} \triangleq \left\{(\x,\z): \log(\wzxn(\z|\x)/Q_0^{\otimes n}(\z)) < \tau\sqrt{n} \right\}.$$ Then, by partitioning $(\x,\z)$ into $(\x,\z) \in \B_{\tau}$ and $(\x,\z) \notin \B_{\tau}$, the term in~\eqref{eq:combine1} can be expressed as  
\begin{align}
& \sum_{(\x,\z) \in \B_{\tau}}  \ppmx(\x)\wzxn(\z|\x) \log\left(1  + \frac{\wzxn(\z|\x)}{|\M|N Q_0^{\otimes n}(\z)} \frac{Q_0^{\otimes n}(\z)}{\ppmz(\z)} \right) + \sum_{(\x,\z) \notin \B_{\tau}} \ppmx(\x)\wzxn(\z|\x)\log\left(1 + \frac{\wzxn(\z|\x)}{|\M|N \ppmz(\z)} \right). \label{eq:sj}
\end{align}	
The first term of~\eqref{eq:sj} is bounded from above by
\begin{align}
\frac{e^{\tau\sqrt{n}}}{|\M|N}\sum_{(\x,\z) \in \B_{\tau}}\ppmx(\x)\wzxn(\z|\x)\frac{Q_0^{\otimes n}(\z)}{\ppmz(\z)} \le \frac{e^{\tau\sqrt{n}}}{|\M|N}, \label{eq:glass}
\end{align}
and the second term of~\eqref{eq:sj} is bounded from above by 
\begin{align}
\log\left(1 + \frac{1}{(|\M|N)\min_{\z:\ppmz(\z)>0}\ppmz(\z)} \right) \times\PP_{\ppmx\wzxn}\left(\log\frac{\wzxn(\Z|\X)}{Q_0^{\otimes n}(\Z)} \ge \tau\sqrt{n} \right). \label{eq:combine2}
\end{align}

We now consider the first term in~\eqref{eq:combine2}. Recall that it is assumed $Q_1 \ll Q_0$, and without loss of generality, we assume there does not exist a symbol $z$ such that $Q_1(z) = Q_0(z) = 0$ . Let $\mathcal{Z}'\triangleq \{z\in\mathcal{Z}: Q_1(z) = 0, Q_0(z) > 0 \}$ be the subset of symbols that are impossible to be induced by the input symbol $X = 1$. Let $\mathcal{I}(\z) \triangleq \{j \in [1:n]: z_j \in \mathcal{Z}' \}$ be the set of locations such that the corresponding elements belong to $\mathcal{Z}'$. Note that if $\z$ satisfies $\ppmz(\z) > 0$, the cardinality of $\mathcal{I}(\z)$ must satisfy $|\mathcal{I}(\z)| \le n - l$. For any $\z$ such that $\ppmz(\z) > 0$, one can always find an $\widetilde{\x}$ such that $\ppmx(\widetilde{\x}) > 0$ and $\mathcal{I}(\z) \cap \text{supp}(\widetilde{\x}) = \emptyset$; thus 
\begin{align}
\wzxn(\z|\widetilde{\x}) = \left(\prod_{j:\widetilde{x}_j = 1} P_1(z_j) \right) \left(\prod_{j:\widetilde{x}_j = 0} P_0(z_j) \right) \ge (\mu_1)^{l} (\mu_0)^{n - l} \ge \widetilde{\mu}^n, \label{eq:mu3}
\end{align}    
where $\mu_0 = \min_{z: Q_0(z)>0}Q_0(z)$, $\mu_1 = \min_{z: Q_1(z)>0}Q_1(z)$, and $\widetilde{\mu} = \min\{\mu_0, \mu_1\}$. Then, we have 
\begin{align}
(|\M|N) \min_{\z: \ppmz(\z)>0}\ppmz(\z) 
&= (|\M|N) \min_{\z: \ppmz(\z)>0} \sum_{\x} \ppmx(\x)\wzxn(\z|\x) \notag \\
&\ge \frac{(|\M|N)}{w^l} \min_{\z:\ppmz(\z)>0}  \sum_{\x: \ppmx(\x)>0} \wzxn(\z|\x) \label{eq:mu0} \\
&\ge \min_{\z:\ppmz(\z)>0}  \sum_{\x: \ppmx(\x)>0} \wzxn(\z|\x) \label{eq:mu1}\\*
&\ge \widetilde{\mu}^n. \label{eq:mu2}
\end{align}
where~\eqref{eq:mu1} holds since $|\M| = \exp\{e^{R\sqrt{n}} \}$ and $w^l = \exp\{\Theta(\sqrt{n}\log n) \}$, and~\eqref{eq:mu2} is true since we know from~\eqref{eq:mu3} that for every $\z$ such that $\ppmz(\z) > 0$, one can find an $\widetilde{\x}$ with $\ppmx(\widetilde{\x}) > 0$ to ensure 
\begin{align}
\sum_{\x: \ppmx(\x)>0}  \wzxn(\z|\x) \ge \wzxn(\z|\widetilde{\x}) \ge \widetilde{\mu}^n. \notag
\end{align}
Thus, we have 
\begin{align}
\log\left(1+ \frac{1}{(|\M|N)\min_{\z:\ppmz(\z)>0}\ppmz(\z) }\right) \le \log\left(1+ \widetilde{\mu}^n \right) \le \log\left((1+ \widetilde{\mu})^n \right) 
= n \log\left(1+ \widetilde{\mu} \right). \label{eq:combine3}
\end{align}

It then remains to consider the second term in~\eqref{eq:combine2}. Note that 
\begin{align}
\PP_{\ppmx\wzxn}\left(\log\frac{\wzxn(\Z|\X)}{Q_0^{\otimes n}(\Z)} \ge \tau\sqrt{n} \right) 
&= \sum_{\x} \ppmx(\x) \sum_{\z} \wzxn(\z|\x) \mathbbm{1}\left\{\log\frac{\wzxn(\z|\x)}{Q_0^{\otimes n}(\z)} \ge \tau\sqrt{n} \right\} \notag\\
&= \sum_{\z} \wzxn(\z|\x^*) \mathbbm{1}\left\{\sum_{j=1}^{n}\log\frac{\wzx(z_i|x^*_j)}{Q_0(z_j)} \ge \tau\sqrt{n} \right\} \label{eq:show1}\\
&= \PP_{Q_1^{\otimes l}}\left(\sum_{j=1}^{l}\log\frac{Q_1(Z_{(j-1)w+1})}{Q_0(Z_{(j-1)w+1})} \ge \tau\sqrt{n} \right), \label{eq:show2}
\end{align}
where~\eqref{eq:show1} is due to symmetry and recall that $\x^*$ is the weight-$l$ vector such that $x^*_{(j-1)w+1} = 1$ for $j \in [1:l]$. By noting that $\E(\sum_{j=1}^{l}\log\frac{Q_1(Z_{(j-1)w+1})}{Q_0(Z_{(j-1)w+1})}) = l\DD(Q_1 \Vert Q_0)$ and $\tau\sqrt{n} \triangleq 2l\DD(Q_1 \Vert Q_0)$, applying Hoeffding's inequality yields 
\begin{align}
\PP_{Q_1^{\otimes l}}\left(\sum_{j=1}^{l}\log\frac{Q_1(Z_{(j-1)w+1})}{Q_0(Z_{(j-1)w+1})} \ge \tau\sqrt{n} \right) \le 2e^{-c_4\sqrt{n}}  \label{eq:combine4}
\end{align}
for some constant $c_4 > 0$. By combining Eqns.~\eqref{eq:glass},~\eqref{eq:combine3} and~\eqref{eq:combine4}, we complete the proof of Lemma~\ref{lemma:6}.

\section{Proof of Lemma~\ref{lemma:CR}} \label{appendix:resolvability}
Let $\zeta > 0$, and we decompose $\widetilde{P}_{\Y}$ into two sub-distributions $\widetilde{P}^{(1)}_{\Y}$ and $\widetilde{P}^{(2)}_{\Y}$ such that
\begin{align}
&\widetilde{P}^{(1)}_{\Y}(\y) \triangleq \frac{1}{K}\sum_{i=1}^K \wyxn(\y|\x_i) \mathbbm{1}\left\{\log\frac{\wyxn(\y|\x_i)}{P'_{\Y}(\y)} > \zeta \right\}, \notag\\
&\widetilde{P}^{(2)}_{\Y}(\y) \triangleq \frac{1}{K}\sum_{i=1}^K \wyxn(\y|\x_i) \mathbbm{1}\left\{\log\frac{\wyxn(\y|\x_i)}{P'_{\Y}(\y)} \le \zeta \right\}.\notag
\end{align}
By noting that $P_{\Y}(\y) = \E(\widetilde{P}_{\Y}(\y))$, where the expectation is over the random generation of $\{\x_1, \ldots, \x_K \}$, we have 
\begin{align}
\E\left(\V\left(P_{\Y}, \widetilde{P}_{\Y}  \right) \right) 
&= \frac{1}{2}\E\left(\sum_{\y}\left|\E\left(\widetilde{P}_{\Y}(\y)\right) - \widetilde{P}_{\Y}(\y)  \right| \right) \notag\\
&\le \frac{1}{2}\E\left(\sum_{\y}\left|\E\left(\widetilde{P}^{(1)}_{\Y}(\y)\right) - \widetilde{P}^{(1)}_{\Y}(\y)  \right| \right) + \frac{1}{2}\E\left(\sum_{\y}\left|\E\left(\widetilde{P}^{(2)}_{\Y}(\y)\right) - \widetilde{P}^{(2)}_{\Y}(\y)  \right| \right). \label{eq:two}
\end{align}
The first term of~\eqref{eq:two} is bounded from above by 
\begin{align}
\sum_{\y}\E\left(\widetilde{P}^{(1)}_{\Y}(\y)\right)
&= \frac{1}{K}\sum_{i=1}^K \sum_{\y}\sum_{\x_i}P_{\X}(\x_i) \wyxn(\y|\x_i) \mathbbm{1}\left\{\log\frac{\wyxn(\y|\x_i)}{P'_{\Y}(\y)} > \zeta \right\} \notag\\
&= \PP_{P_{\X}\wyxn}\left(\log\frac{\wyxn(\Y|\X)}{P'_{\Y}(\Y)} > \zeta \right). \notag
\end{align} 
By applying Jensen's inequality, the second term of~\eqref{eq:two} is bounded from above by
\begin{align}
\frac{1}{2}\sum_{\y}\E\left(\sqrt{\left(\E\left(\widetilde{P}^{(2)}_{\Y}(\y)\right) - \widetilde{P}^{(2)}_{\Y}(\y)\right)^2} \right) &\le \frac{1}{2}\sum_{\y}\sqrt{\E\left[\left(\E\left(\widetilde{P}^{(2)}_{\Y}(\y)\right) - \widetilde{P}^{(2)}_{\Y}(\y)\right)^2\right]} \notag\\
&= \frac{1}{2}\sum_{\y} \sqrt{\text{Var}\left( \widetilde{P}^{(2)}_{\Y}(\y) \right) }, \label{eq:further}
\end{align}
and one can further show that
\begin{align}
\text{Var}\left( \widetilde{P}^{(2)}_{\Y}(\y) \right)
&= \frac{1}{K^2} \sum_{i=1}^K \text{Var}\left( \wyxn(\y|\X_i) \mathbbm{1}\left\{\log\frac{\wyxn(\y|\X_i)}{P'_{\Y}(\y)} \le \zeta \right\} \right) \notag\\
&\le \frac{1}{K^2} \sum_{i=1}^K \E\left( \wyxn(\y|\X_i)^2 \mathbbm{1}\left\{\log\frac{\wyxn(\y|\X_i)}{P'_{\Y}(\y)} \le \zeta \right\} \right) \notag\\
&\le \frac{1}{K^2} \sum_{i=1}^K \E\left( \wyxn(\y|\X_i) e^{\zeta} P'_{\Y}(\y) \right) \notag\\*
&= \frac{e^{\zeta}}{K} P'_{\Y}(\y) P_{\Y}(\y).\notag
\end{align}
By using the arithmetic-geometric mean inequality, we see that~\eqref{eq:further} is further bounded from above as 
\begin{align}
\frac{1}{2}\sum_{\y} \sqrt{\text{Var}\left( \widetilde{P}^{(2)}_{\Y}(\y) \right) } &\le \frac{1}{2}\sum_{\y} \sqrt{\frac{e^{\zeta}}{K}  P'_{\Y}(\y) P_{\Y}(\y)} \notag\\
&\le \frac{1}{2}\sqrt{\frac{e^{\zeta}}{K}} \sum_{\y} \frac{P'_{\Y}(\y) + P_{\Y}(\y)}{2}  \notag\\
&= \frac{1}{2}\sqrt{\frac{e^{\zeta}}{K}}.\notag
\end{align}

\ifCLASSOPTIONcaptionsoff
  \newpage
\fi

\bibliographystyle{IEEEtran}
\bibliography{reference}

\end{document}